\documentclass[conference,twocolumn]{IEEEtran}
\usepackage[ruled,vlined]{algorithm2e}
\usepackage{amsfonts}
\usepackage{amssymb}

\usepackage{cancel}
\usepackage{ClearSans}

\usepackage{epsfig}

\usepackage{glossaries}

\usepackage{ifthen}

\usepackage{mathrsfs}
\usepackage{mathtools}

\usepackage{scalerel}
\usepackage{subcaption}

\usepackage{tikz,pgfplots}
\pgfplotsset{compat=newest}

\usepackage[normalem]{ulem}

\usepackage{cite}

\usepackage{enumitem}
\usepackage[draft]{hyperref}

\usepackage{dsfont}

\newcommand{\diag}{\ensuremath{\text{diag} }}

\newcommand{\Exp}[2][]{\ensuremath{\mathbb{E}_{#1}\left[#2 \right]}} 

\newcommand{\eqann}[2]{\overset{\mathclap{(\text{#2})}}{#1}} 
\newcommand{\eqannref}[1]{$(\text{#1})$}
\newcommand{\bv}[1]{\mathbf{#1}} 
\newcommand{\rv}[1]{\mathsf{#1}} 

\newcommand{\minus}{\scalebox{0.75}[1.0]{\( - \)}}

\newcommand{\pmf}[1]{\ensuremath{\mathsf{P}_{#1}}}

\newcommand{\set}[1]{\mathcal{#1}} 

\newtheorem{theorem}{Theorem}

\newtheorem{remark}{Remark}

\newtheorem{lemma}{Lemma}

\newenvironment{proof}[1][Proof]{\noindent\textbf{#1.} }{\ \rule{0.5em}{0.5em}}

\newacronym{awgn}{AWGN}{\textit{additive white Gaussian noise}}

\newacronym{bac}{BAC}{Binary Asymmetric Channel}

\newacronym{bdsib}{BDSIB}{Binary Double-Sided Information-Bottleneck}
\newacronym{bec}{BEC}{\textit{binary erasure channel}}

\newacronym{bms}{BMS}{Binary Memoryless Symmetric}

\newacronym{bs}{BS}{base station}

\newacronym{bsc}{BSC}{\textit{binary symmetric channel}}
\newacronym{bscs}{BSCs}{\textit{binary symmetric channels}}

\newacronym{ceb}{CEB}{\textit{conditional entropy bound}}

\newacronym{comib}{COMIB}{\textit{compound information bottleneck}}

\newacronym{cr}{CR}{\textit{common reconstruction}}

\newacronym{cran}{C-RAN}{cloud radio access network}

\newacronym{cp}{CP}{Central Proccesor}

\newacronym{dmc}{DMC}{\textit{discrete memoryless channel}}

\newacronym{dmmac}{DM-MAC}{Discrete Memoryless Multiple Access Channel}

\newacronym{dms}{DMS}{Discrete Memoryless Source}

\newacronym{dpi}{DPI}{Data Proccesing Inequality}

\newacronym{dsbs}{DSBS}{\textit{doubly symmetric binary source}}

\newacronym{dsib}{DSIB}{Double-Sided Information-Bottleneck}

\newacronym{epi}{EPI}{Entropy Power Inequality}

\newacronym{gdsib}{GDSIB}{Gaussian Double-Sided Information Bottleneck}

\newacronym{gp}{GP}{Gelf'and-Pinsker}

\newacronym{ib}{IB}{\textit{information bottleneck}}
\newacronym{iid}{i.i.d.}{independent and identically distributed}

\newacronym{infcomb}{IC}{\textit{information combining}}

\newacronym{kld}{KLD}{\textit{Kullback–Leibler divergence}}

\newacronym{lhs}{LHS}{Left Hand Side}

\newacronym{sgdsib}{SGDSIB}{Scalar Gaussian Double-Sided Information-Bottleneck}

\newacronym{mac}{MAC}{multiple access channel}

\newacronym{mgl}{MGL}{Mrs. Gerber's lemma}

\newacronym{mi}{MI}{mutual information}

\newacronym{mmse}{MMSE}{minimum mean squared error}

\newacronym{mu}{MU}{Mobile User}

\newacronym{oblib}{OBLIB}{Oblivious Information Bottleneck}

\newacronym{pf}{PF}{\textit{Privacy Funnel}}

\newacronym{pmf}{PMF}{probability mass function}

\newacronym{ssib}{SSIB}{Single-Sided Information Bottleneck}
\newacronym{sgssib}{SGSSIB}{Scalar Gaussian Single-Sided Information Bottleneck function}
\newacronym{sawgnssib}{SAWGNSSIB}{Scalar AWGN Single-Sided Information Bottleneck function}

\newacronym{snr}{SNR}{Signal-to-Noise Ratio}

\newacronym{svd}{SVD}{Singular Value Decomposition}

\newacronym{tibo}{TIBO}{Ternary-Input Binary-Output }

\newacronym{tito}{TITO}{Ternary-Input Ternary-Output }

\newacronym{qiqo}{QIQO}{Quaternary-Input Quaternary-Output }

\newacronym{rhs}{RHS}{right hand side}

\newacronym{rv}{RV}{Random Variable}

\newacronym{rssib}{RSSIB}{Reversed Single-Sided Information Bottleneck function}
\newacronym{rsawgnssib}{RSAWGNSSIB}{Reversed Scalar AWGN Single-Sided Information Bottleneck function}

\newacronym{tv}{TV}{\textit{Total Variation}}

\newacronym{wlog}{WLOG}{long}

\newacronym{wtc}{WTC}{Wiretap Channel}
\newacronym{psd}{PSD}{positive semi-definite}
\newacronym{wss}{WSS}{wide sense stationary}
\newacronym{dpcm}{DPCM}{\textit{differential pulse-code modulation}}

\newacronym{isi}{ISI}{inter-symbol interference}
\newacronym{vq}{VQ}{Vector Quantizer}

\usepackage{setspace}
\PassOptionsToPackage{normalem}{ulem}
\usepackage{ulem}

\newboolean{fullver}
\setboolean{fullver}{true} 

\title{On Information Bottleneck for Gaussian Processes}

\author{
	\IEEEauthorblockN{Michael Dikshtein, Nir Weinberger,  and Shlomo Shamai (Shitz)}
	
	\IEEEauthorblockA{Department of Electrical and Computer Engineering, Technion, Haifa 3200003, Israel}
	
	\IEEEauthorblockA{ 		
		Email: \{michaeldic@campus., nirwein@, sshlomo@ee.\}technion.ac.il
}
}	

\begin{document}
	\maketitle
	
	\begin{abstract}
	The \acrfull{ib} problem of jointly stationary Gaussian sources is considered. A water-filling solution for the \acrshort{ib} rate is given in terms of its \acrshort{snr} spectrum and whose rate is attained via frequency domain test-channel realization. A time-domain realization of the \acrshort{ib} rate, based on linear prediction, is also proposed, which lends itself to an efficient implementation of the corresponding remote source-coding problem. A compound version of the problem is addressed, in which the joint distribution of the source is not precisely specified but rather in terms of a lower bound on the guaranteed mutual information. It is proved that a white \acrshort{snr} spectrum is optimal for this setting.
	\end{abstract}
	
	\section{Introduction}	
	\begin{figure}[b]
	\centering
	\begin{tikzpicture}[thick,scale=0.87, every node/.style={scale=0.9}]
		\coordinate  (inX) at (0,0);
		\node[rectangle, draw = black, minimum width = 1.5cm, minimum height = 1.2cm, align = center]  (source) at (2,0) {Source \\ $ \pmf{\rv{X}  \rv{Y}} $};
		\node[rectangle, draw = black, minimum width = 1cm, minimum height = 1.2cm, text width = 1.75cm, align = center]  (enc) at (5,0) {Encoder};
		\node[rectangle, draw = black, minimum width = 1cm, minimum height = 1.2cm, text width = 1.75cm, align = center]  (dec) at (8,0) {Decoder};
		\node (out) at (9.5,0) {};

		\draw[->] (source)  --  node [above]  {$ \rv{X}^n $} (inX) ;
		\draw[->] (source) -- node[above] {$ \rv{Y}^n $}  (enc);
		\draw[->] (enc) --  node[above] {$ \rv{M}  $} (dec);
		\draw[->] (dec) -- (out) node[right] {$ \rv{Z}^n $};
		
		\end{tikzpicture}
	\caption{Block diagram of Remote Source Coding.}
	\label{figure:remote_source_coding_diagram}
\end{figure}
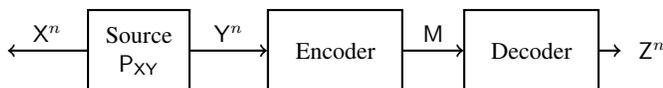

It is by now well established that the \acrfull{ib} method \cite{Tishby1999} plays a central role in information theory,  machine learning, and various other fields. In its basic form, it introduced an alternative conceptual approach to the problem of lossy compression. In classic rate-distortion problems \cite{Cover2006}, a distortion measure must be provided in order to calculate the respective rate-distortion function and characterize the single-letter compression strategy. However, in practice, determining the distortion measure may be a challenging problem on its own. One motivating example is speech compression, in which it is difficult to identify the prominent features of the vocal signal that affect the quality of the reconstructed audio. Also, one may argue that the transcript of the speech signal has more information than the features of the analog waveform. The \acrshort{ib} framework tackles those issues by introducing an additional, informative variable to the problem that acts as the appropriate labeling of the data to be compressed. Thus, \acrshort{ib} provides a natural fidelity measure in cases where such cannot be defined or does not exist; an essential feature for modern source coding problems. We also mention that, as is well known, the \acrshort{ib} is the information-theoretic solution to a remote source coding problem \cite{Dobrushin1962,Wolf1970}, when the  log-loss is chosen as the distortion measure. Nonetheless, the setting behind remote-source coding is fundamentally different from \acrshort{ib}, despite the formal similarity between the resulting  optimization problems. 

In general, the \acrshort{ib} function is a non-convex optimization problem and does not have a closed-form solution. However, it can be analytically solved in some canonical scenarios, mainly the \acrfull{dsbs} in \cite{Zaidi2020}, and a jointly Gaussian vector source in \cite{Chechik2005}. In all other cases, as proposed in \cite{Tishby1999}, a Blahut-Arimoto type alternating minimization algorithm \cite{Blahut1972,Arimoto1972}, is typically used for obtaining a solution. This algorithm was shown to converge to a local stationary point but has no global convergence guarantees.

This work addresses the bivariate Gaussian setting and replaces the finite length vectors formulation with that of a stationary Gaussian random process. Considering random processes instead of finite length vectors gives the problem a more natural and practical flavor and is motivated by the communication and signal-processing problems, in which the receiver sequentially processes the samples. It has also rooted in remote source coding and signal-denoising applications. This wide-sense stationary setting is usually approached via frequency-domain methods, leading to linear time-invariant system implementation. In this paper, we will also propose a prediction-based scheme that implements the compression phase using a time-domain single-letter, sequential processing.

\subsection{Problem Formulation}
Consider the discrete-time \acrshort{ib} model illustrated in \autoref{figure:remote_source_coding_diagram}. The real valued bivariate source
\begin{equation} \label{eq:bivariate_source}
  (\{\rv{X}_t\},\{\rv{Y}_t\}) = \dots, (\rv{X}_{-1},\rv{Y}_{-1}), (\rv{X}_{0},\rv{Y}_{0}), (\rv{X}_{1},\rv{Y}_{1}),  \dots  ,
\end{equation}
is a bivariate stationary Gaussian random process, with marginal power spectral densities $ S_{\rv{X}}(f)  $, $ S_{\rv{Y}} (f) $, and cross-power spectrum $ S_{\rv{X}\rv{Y}}(f) $. The encoder observes $\rv{Y}^n$ and maps it to a compressed representation with an index $\rv{M} \in [1:2^{nC}]$. The decoder recovers $\rv{Z}^n$ from $\rv{M}$. The compression strategy $\pmf{\rv{Z}|\rv{Y}}$ is optimized to maximize the normalized mutual information between $\rv{X}^n$ and $\rv{Z}^n$.

The considered jointly Gaussian bivariate stationary source can be equivalently represented by linear time invariant filters \cite[Thm. 4.5.5]{Berger1971}, i.e.,
	$\rv{Y}_n = h_n * \rv{X}_n + \rv{W}_n$;
where $*$ stands for convolution operator, $h_n$ is the impulse response of the linear system with transfer function
\begin{equation}
	H(f) = \frac{S_{\rv{X}\rv{Y}}(f)}{S_{\rv{X}}(f)},
\end{equation}
and $\rv{W}_n$ is an additive colored Gaussian noise with power spectrum
\begin{equation}
	S_{\rv{W}}(f) = S_{\rv{Y}}(f) - \frac{|S_{\rv{X}\rv{Y}}(f)|^2}{S_{\rv{X}}(f)}.
\end{equation}

The \acrfull{ib} rate of a bivariate stationary source with memory, which is given as a limit of normalized mutual information associated with vectors of source samples,  can be written as
\begin{equation} \label{eq:ib_rate_definion}
	\mathcal{R}^{ib}(C) = \lim_{ n \rightarrow \infty } R_n^{ib}(C;\pmf{\rv{X}^n \rv{Y}^n}
	),
\end{equation}
where
\begin{equation} \label{eq:normalized_ib_function}
	\begin{aligned}
		& R_n^{ib}(C;\pmf{\rv{X}^n \rv{Y}^n}) =
		&& \sup_{\pmf{\rv{Z}^n|\rv{Y}^n}} 
		&&  \frac{1}{n} I(\rv{X}^n;\rv{Z}^n) \\
		& && \text{subject to}
		&& \frac{1}{n} I(\rv{Y}^n;\rv{Z}^n) \leq C
	\end{aligned}
\end{equation}
is the normalized \acrshort{ib} function for random vectors $ (\rv{X}^n,\rv{Y}^n,\rv{Z}^n) $. We will usually abbreviate $ R_n^{ib}(C;\pmf{\rv{X}^n \rv{Y}^n}) $ as $ R_n^{ib}(C) $. The channel $ \pmf{\rv{Z}^n|\rv{Y}^n} $ which achieves this supremum subject to the bottleneck constraint is termed an \emph{optimal test channel}.

\paragraph*{Contributions}
Before understanding complex models such as DNN, one must first understand the canonical models, such as the Gaussian one. This is also a common theoretical approach in machine learning. For example, there are many works on two-layer neural networks, or even linear networks (which are in some sense trivial since a concatenation of linear operations is just a linear operation on its own). The IB method provides a natural fidelity measure for modern source coding applications. The standard rate-distortion problem for Gaussian processes with least-squares loss has been considered in [14]. Thus, our problem extends previous works by considering a more general distortion measure. Furthermore, stochastic colored process stands as a good model for a general source with continuous stream of output symbols, i.e. audio, video, sensor output. Gaussian processes are a simple model for such processes which are usually manageable for analytical consideration.

We first state a ``water filling" solution for the \acrshort{ib} function of a stationary bivariate Gaussian source in terms of its \acrshort{snr} spectrum. Since \acrshort{ib} is essentially a remote source coding with log-loss distortion \cite{Dobrushin1962,Wolf1970}, the resulting formula is similar in spirit to the capacity formula of power-constrained \acrfull{isi} channel with Gaussian noise \cite{Hirt1988,Shamai1991}, and the rate-distortion function of a stationary Gaussian source \cite{Berger1971}. 
The result above is derived from asymptotic quantities of mutual information between Gaussian random vectors in the frequency domain \cite{Gray2006}. We next translate these vector mutual information measures to scalar ones via linear prediction. This parallels a result which states that the capacity of the \acrshort{isi} channel is equal to the single letter mutual information over a slicer embedded in a decision-feedback noise-prediction loop \cite{Cioffi1995}, and similarly, that the rate-distortion function equals the single letter mutual information over an \acrfull{awgn} channel embedded in a source prediction loop \cite{Zamir2008}.
We show that a parallel result holds for the \acrshort{ib} rate $ \mathcal{R}^{ib}(C) $, which is equal to the scalar mutual information over an \acrshort{awgn} channel embedded in a source prediction loop, as shown in \autoref{figure:predictive_test_channel}. This result implies that $ \mathcal{R}^{ib}(C) $ can essentially be realized sequentially. 
Finally, we consider a compound version of the \acrshort{ib} problem with Gaussian processes, in which the cross-spectrum of the processes $(\{\rv{X}_t\},\{\rv{Y}_t\})$ is not fully specified, and it is only known that the mutual information rate $\mathcal{I}(\{\rv{X}_t\};\{\rv{Y}_t\})\geq C_1$ for some $C_1$. In general, this problem is motivated by the uncertainty of the correlation between $ \{\rv{X}_t\}$ and $\{\rv{Y}_t\}$, and we refer the reader to \cite{Dikshtein2022a} for a more elaborate discussion. Here we provide an explicit solution for this scenario.

\ifthenelse{\boolean{fullver}}{}{
Omitted proofs and other details are in the full version of this paper \cite{Dikshtein2022b}. \looseness=-1}

\subsection{Related Work}
The vector Gaussian \acrshort{ib} setting was first considered in \cite{Chechik2005}. It was shown that jointly Gaussian vectors triple $ \bv{X} \rightarrow \bv{Y} \rightarrow \bv{Z} $ is optimal \cite{Globerson2004},  and a closed-form formula for the \acrshort{ib} curve was obtained. Linear operations can typically obtain optimal solutions for the Gaussian setting. However, channel output compression subject to squared-error distortion was shown to be sub-optimal to the Gaussian \acrshort{ib} setting \cite{Winkelbauer2014}. The latter issue was solved in \cite{Meidlinger2014} by adding a prefilter prior the rate-distortion block. The discrete-time Gaussian process setting was considered in \cite{Meidlinger2014} and \cite{Homri2018}. 

For the jointly stationary bivariate source $(\{\rv{X}_t\},\{\rv{Y}_t\})$, the \acrfull{pf} rate \cite{Makhdoumi2014} is defined by
\begin{equation} \label{eq:pf_rate_definion}
	\mathcal{R}^{pf}(C) = \lim_{ n \rightarrow \infty } R_n^{pf}(C;\pmf{\rv{Y}^n \rv{Z}^n}
	),
\end{equation}
where
\begin{equation} \label{eq:normalized_pf_function}
	\begin{aligned}
		& R_n^{pf}(C;\pmf{\rv{Y}^n \rv{Z}^n}) =
		&& \inf_{\pmf{\rv{Z}^n|\rv{Y}^n}} 
		&&  \frac{1}{n} I(\rv{X}^n;\rv{Z}^n) \\
		& && \text{subject to}
		&& \frac{1}{n} I(\rv{X}^n;\rv{Y}^n) \geq C
	\end{aligned}
\end{equation}
is the normalized \acrshort{pf} function for random vectors $ (\rv{X}^n,\rv{Y}^n,\rv{Z}^n) $. 


	\section{Information Bottleneck for Gaussian Processes}

    In this section, we review and state results on \acrshort{ib} for jointly Gaussian random vectors and jointly Gaussian random processes.
    \subsection{Review: \acrshort{ib} for jointly Gaussian pair $     (\bv{X},\bv{Y}) $}
    We begin with a brief review of the basic setting in  which $ (\bv{X},\bv{Y}) $ in \eqref{eq:normalized_ib_function} are jointly Gaussian random vectors.  Let $ \bv{X} \sim \mathcal{N}(\bv{0},\Sigma_{\bv{X}}) $, $ \bv{Y} \sim \mathcal{N}(\bv{0},\Sigma_{\bv{Y}}) $ with cross-covariance matrix $ \Sigma_{\bv{X}\bv{Y}} $. Those bivariate random Gaussian vectors can equivalently represented using a linear additive-noise form, i.e.,
\begin{equation}
	\bv{Y} = H \bv{X} + \bv{W},
\end{equation}
where $ H = \Sigma_{\bv{X}\bv{Y}} \Sigma_{\bv{X}}^{-1} $ and $ \Sigma_{\bv{W}} = \Sigma_{\bv{Y}} - \Sigma_{\bv{X}\bv{Y}}^T \Sigma_{\bv{X}}^{-1} \Sigma_{\bv{X}\bv{Y}}  $. Let $ O \Gamma O^T $ be the \acrfull{svd} of  the \acrfull{snr} covariance matrix $ \Sigma_{snr} $, defined by
\begin{equation}
	\Sigma_{snr} \triangleq \Sigma_{\bv{W}}^{-1/2} H \Sigma_{\bv{X}} H^T \Sigma_{\bv{W}}^{-1/2},
\end{equation}
with $ \Gamma = \diag \{\gamma_i\}_{i=1}^n  $.
The vector Gaussian \acrshort{ib} function was first determined in \cite{Chechik2005} (see also  \cite{Meidlinger2014} and \cite[Sec. 2]{Goldfeld2020}, where the trade-off parameter $ \beta $ is replaced here by $ 1+1/\theta $). We next present this result with a structure that resembles the classical rate-distortion function for Gaussian sources with memory \cite[Thm. 4.5.3]{Berger1971}, where $ \theta $ plays the role of the water-filling level -- First $ \theta $ is chosen to satisfy the bottleneck constraint $ C $, and then the information rate is calculated. The importance of this form of presentation is that it will be used in the next section to obtain realizations in the time domain. 
\begin{lemma} \label{lemma:ib_jointly_gaussian_vectors}
	The normalized \acrshort{ib} rate \eqref{eq:normalized_ib_function} is achieved by a jointly Gaussian triple $(\bv{X},\bv{Y},\rv{Z})$ and is given by
	\begin{equation}
		R_n^{ib}(C) 
		=
		\frac{1}{2n} \sum_{i=1}^n
		\log
		\left[
		\frac{1+\gamma_i}{1+\theta}
		\right]^+,
	\end{equation}
	where  $\theta$ is  the water filling level chosen such that
	\begin{equation}
		\frac{1}{2n} \sum_{i=1}^N  \log\left[ \frac{ \gamma_i }{\theta}\right]^+ = C,
	\end{equation}
	with $ [x]^+ \triangleq \max \{1,x\} $,
	and $ \{\gamma_i\}_{i=1}^n $ are the eigenvalues of the \acrshort{snr} covariance matrix $ \Sigma_{snr} $.
\end{lemma}


	\subsection{\acrshort{ib} for jointly Gaussian Processes}
    Equipped with the form of the \acrshort{ib} function in \autoref{lemma:ib_jointly_gaussian_vectors}, the solution to the discrete-time processes setting  follows from Szeg\"o's theorem \cite[Thm. 4.2]{Gray2006}.  While a related result with filtered observation was first published in \cite{Meidlinger2014}, in our approach such assumption is not required.
\begin{theorem} \label{theorem:ib_gaussian_process}
	The \acrshort{ib} rate \eqref{eq:ib_rate_definion} for Gaussian random processes is given by
	\begin{equation}
		\mathcal{R}^{ib}(C) = \frac{1}{2} \int_{-1/2}^{1/2}  \log \left[ \frac{1+\Gamma(f)}{1+\theta}\right]^+ \mathrm{d} f,
	\end{equation}
	where we choose the \textit{water level}  $ \theta $ so that the total rate is $ C $,
	\begin{equation}
		C = \frac{1}{2}  \int_{-1/2}^{1/2} \log \left[ \frac{\Gamma(f)}{\theta} \right]^+ \mathrm{d} f,
	\end{equation}
	and $ \Gamma(f) $ is the \acrshort{snr} spectrum, defined by
	\begin{equation}
		\Gamma(f) = \frac{|H(f)|^2  S_{\rv{X}}(f)}{S_{\rv{W}}(f)}= \frac{|S_{\rv{X}\rv{Y}}(f)|^2}{S_{\rv{X}}(f) S_{\rv{Y}}(f) - |S_{\rv{X}\rv{Y}}(f)|^2}.
	\end{equation}
\end{theorem}

We may define the \textit{distortion spectrum} \cite{Berger1971} as  
	\begin{equation} \label{eq:prediction_Dtheta_definition}
		D_{\theta}(f) = 
		\begin{cases}
			\theta, & \Gamma(f) > \theta \\
			\Gamma(f), & \text{otherwise}
		\end{cases}
	\end{equation}
	and then \autoref{theorem:ib_gaussian_process} can be equivalently stated as
	\begin{equation}
		\mathcal{R}^{ib}(C) = \frac{1}{2} \int_{-1/2}^{1/2}  \log \left[ \frac{1+\Gamma(f)}{1+D_{\theta}(f)}\right] \mathrm{d} f,
	\end{equation}
	and 
	\begin{equation}
		C = \frac{1}{2} \int_{-1/2}^{1/2}  \log \left[ \frac{\Gamma(f)}{D_{\theta}(f)} \right] \mathrm{d} f.
	\end{equation}

As immediate consequence of \autoref{theorem:ib_gaussian_process} is that the  optimal channel from $\rv{Y}_n$ to $\rv{Z}_n$ can be described in a linear \acrshort{awgn} form
\begin{equation} \label{eq:ib_rate_linear_awgn_form}
	\rv{Z}_n = h_{2,n} * (h_{1,n} * g_n * \omega_n *  \rv{Y}_n + \rv{N}_n),
\end{equation}
where $ \omega_n $, $ g_n $, $ h_{1,n} $ and $ h_{2,n} $ are impulse responses of a noise-whitening filter, shaping filter, a suitable prefilter and postfilter respectively, whose absolute squared value frequency responses are given by
\begin{align}
	|\Omega(f) |^2 & = \frac{1}{S_{\rv{W}}(f)}, \quad
	|G(f)|^2  = \frac{\Gamma(f)}{1+ \Gamma(f)}, \\
	|H_1(f)|^2 &= 1- \frac{D(f)}{S_{\rv{Y}}(f)}, \quad
	H_2(f) = H_1^*(f),
\end{align}
and $ \rv{N}_n \sim \mathcal{N}(0,\theta) $. The respective forward channel realization is illustrated in \autoref{figure:forward_channel_realization}.

\begin{remark}
	The absolute frequency response of the shaping filter $ G(f) $ is exactly the square root of the noncausal Wiener filter. Thus, the shaping filter plays the role of denoiser in the spirit of \acrshort{mmse} estimation of $ \{ \rv{X}_t \} $ from $ \{\rv{Y}_t\} $.
\end{remark}



\begin{figure}
	\centering
	\begin{tikzpicture}[scale=0.45, every node/.style={scale=0.6}]
	\coordinate (inY) at (0,0);
	\node[draw = black, text width=2cm,align =center] (whitefilt) at (2,0) {Whitening Filter $ 1/\sqrt{ S_{ \tilde{ \rv{W} } }(f)} $};
	
	\node[draw = black, text width=2cm,align =center] (shapefilt) at (6,0) {Shaping Filter $ G(f) $};
	
	\node[draw = black, text width=2cm,align =center] (prefilt) at (10,0) {Pre-filter $ H_1(f) $};
	
	\node[circle, draw = black] (adder2) at (13,0) {$ \sum $};

	\node[draw = black, text width=2cm,align =center] (postfilt) at (16,0) {Post-filter $ H_2(f) $};
	
	\coordinate (outZ) at (18,0);
		
	\node at (12.5,0.75) {$ + $};
	
	\coordinate (noiseN) at(13,2);
	
	\draw[->] (inY)  node[left]{$ \rv{Y}_n $}-- (whitefilt);
	\draw[->] (whitefilt) -- node[above]{$ \tilde{\rv{Y}}_n $} (shapefilt);
	\draw[->] (shapefilt) -- node[above]{$ \rv{Y}'_n $} (prefilt);
	
	\draw[->] (prefilt) -- node[above]{$ \rv{U}_n $} (adder2);

	\draw[->] (adder2) -- node[above]{$ \rv{V}_n $} (postfilt);
	\draw[->] (postfilt) --  (outZ) node[right]{$ \rv{Z}_n $};

	\draw[->] (noiseN) node[above]{$ \rv{N}_n $} -- (adder2);
	
\end{tikzpicture}
	\caption{Forward channel realization.}
	\label{figure:forward_channel_realization}
\end{figure}
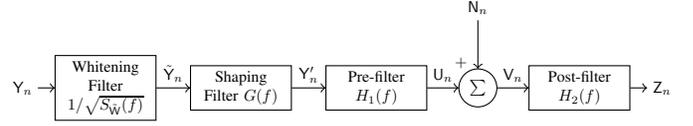

	\section{Remote Source Coding in Time Domain via Prediction}
	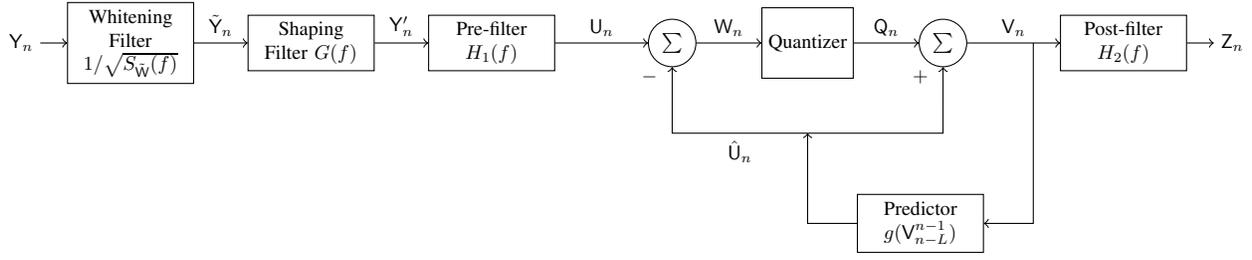
\begin{figure*}[t]
	\centering
	\begin{tikzpicture}[scale=0.6, every node/.style={scale=0.75}]
	\coordinate (inY) at (0,0);
	\node[draw = black, text width=2cm,align =center] (whitefilt) at (2,0) {Whitening Filter $ 1/\sqrt{ S_{ \tilde{ \rv{W} } }(f)} $};
	
	\node[draw = black, text width=2cm,align =center] (shapefilt) at (6,0) {Shaping Filter $ G(f) $};
	
	\node[draw = black, text width=2cm,align =center] (prefilt) at (10,0) {Pre-filter $ H_1(f) $};
	
	\node[circle, draw = black] (adder1) at (14,0) {$ \sum $};
	
	\node[draw=black, minimum height = 1.25cm] (dpcmQ) at (17,0) {Quantizer};
	
	\coordinate (split1) at(17,-2);
	\node[circle, draw = black] (adder3) at (20,0) {$ \sum $};
	
	\node[draw = black, text width=2cm,align =center] (postfilt) at (24,0) {Post-filter $ H_2(f) $};
	
	\coordinate (outZ) at (26,0);
	
	\node[draw = black, text width=2cm,align =center] (predictor) at (19.5,-4) {Predictor $ g(\rv{V}_{n-L}^{n-1}) $};
	
	\node at (13.5,-0.75) {$ - $};
	
	\node at (19.5,-0.75) {$ + $};
	
	
	\coordinate (split2) at(22,0);
	
	\coordinate (noiseN) at(17,2);
	
	\draw[->] (inY)  node[left]{$ \rv{Y}_n $}-- (whitefilt);
	\draw[->] (whitefilt) -- node[above]{$ \tilde{\rv{Y}}_n $} (shapefilt);
	\draw[->] (shapefilt) -- node[above]{$ \rv{Y}'_n $} (prefilt);
	
	\draw[->] (prefilt) -- node[above]{$ \rv{U}_n $} (adder1);
	\draw[->] (adder1) -- node[above]{$ \rv{W}_n $} (dpcmQ);
	\draw[->] (dpcmQ) -- node[above]{$ \rv{Q}_n $} (adder3);
	\draw[->] (adder3) -- node[above]{$ \rv{V}_n $} (postfilt);
	\draw[->] (postfilt) --  (outZ) node[right]{$ \rv{Z}_n $};
	
	\draw[->] (split1) --  node[below]{$ \hat{\rv{U}}_n $} (split1-|adder1) --  (adder1)  ;
	\draw[->] (split1) --   (split1-|adder3) --  (adder3)  ;
	
	\draw[->] (split2) --   (split2|-predictor) --  (predictor);
	
	\draw[->] (predictor) --   (predictor-|split1) --  (split1);
	
	
\end{tikzpicture}
	\caption{A predictive quantization scheme.}
	\label{figure:predictive_test_channel}
\end{figure*}

In this section, we consider the compression system illustrated in \autoref{figure:predictive_test_channel}. This setting is motivated by the sequential structure of \acrfull{dpcm} \cite{Gibson1998}, which was initially proposed in order to compress highly correlated sequences efficiently. The idea behind \acrshort{dpcm} is to translate the encoding of dependent source samples into a series of independent encodings. The latter is accomplished by linear prediction. The input source sample is predicted from previously encoded samples at each instant. The prediction error is encoded by a scalar
quantizer and added to the predicted value to form the new reconstruction.
The prefilter output, denoted $\rv{U}_n $, is fed to the central block, which generates a process $ \rv{V}_n $ according to the following recursion equations:
\begin{align}
	\hat{\rv{U}}_n &= g(\rv{V}_{n-L}^{n-1})  \\
	\rv{W}_n &= \rv{U}_n - \hat{\rv{U}}_n\\
	\rv{Q}_n &= \rv{W}_n + \rv{N}_n \\
	\rv{V}_n & = \hat{\rv{U}}_n  + \rv{Q}_n,
\end{align}
where $\rv{N}_n \sim \mathcal{N}(0,\theta)$ is a zero-mean white Gaussian noise, independent of the input process $ \{\rv{U}_t\} $, whose variance is equal to the water level $ \theta $; and $ g(\cdot) $ is some prediction function for the input $ \rv{U}_n $ given the $ L $ past samples of the output process $ (\rv{V}_{n-L}^{n-1}) $.

The block from $ \rv{U}_n $ to $ \rv{V}_n $ is equivalent to the configuration of \acrshort{dpcm}, with the \acrshort{dpcm} quantizer replaced by the additive Gaussian noise channel $ \rv{Q}_n = \rv{W}_n + \rv{N}_n $. In particular, the recursive structure implies that this block satisfies the well-known ``\acrshort{dpcm} error identity":
\begin{equation}
	\rv{V}_n = \rv{U}_n + (\rv{Q}_n - \rv{W}_n) = \rv{U}_n + \rv{N}_n.
\end{equation}
Namely, the output $ \rv{V}_n $ is a noisy version of $ \rv{U}_n $ passed through the \acrshort{awgn} channel 
$ \rv{V}_n = \rv{W}_n + \rv{N}_n$.
Thus, the system of \autoref{figure:predictive_test_channel} is equivalent to the forward channel implementation of the Gaussian \acrshort{ib} illustrated in \autoref{figure:forward_channel_realization}. 

The prediction function $ g(\cdot) $ is chosen to be linear as in \acrshort{dpcm}, i.e.,
\begin{equation}
	g(\rv{V}_{n-L}^{n-1}) = \sum_{i=1}^L a_i \rv{V}_{n-i},
\end{equation}
where $ \{a_i\}_{i=1}^L $ are chosen to minimize the mean-squared prediction error
\begin{equation} \label{eq:prediction_error_finiteL}
	\sigma_L^2 = \min_{\{a_i\}_{i=1}^n} \Exp{\left[\rv{U}_n - \sum_{i=1}^L a_i \rv{V}_{n-i}\right]^2}.
\end{equation}
Since $ \rv{V}_n $ is noisy version of $ \rv{U}_n $, $ g(\cdot) $ is termed as a ``noisy predictor". If $ \{\rv{U}_t\} $ and $ \{\rv{V}_t\} $ are jointly Gaussian, then the optimal predictor of any order is linear, thus, $ \sigma_L^2 $ is also the \acrshort{mmse} in estimating $ \rv{U}_n $ from $ \rv{V}_{n-L}^{n-1} $. Clearly, this \acrshort{mmse} is nonincreasing with the prediction depth $ L $, and it converges as $ L $ goes to infinity to $\sigma_{\infty}^2 = \lim_{L \rightarrow \infty } \sigma_L^2$, which is the optimal predictor of $ \rv{U}_n $ based on all past observations of $ \rv{V}_n^{-} $. 
\begin{remark}
	Note that $ \sigma_{\infty}^2 $ is closely related to the entropy power. Indeed, let $ \{\rv{X}_t\} $ be a Gaussian source with power spectrum $ S_{\rv{X}}(f) $, then the entropy power is defined by
	\begin{equation}
		P_e(\rv{X}) \triangleq \exp \left(\int_{-1/2}^{1/2} \log \left(S(f)\right) \mathrm{d}f \right).
	\end{equation}
	According to Wiener's spectral-factorization theory, the entropy power is the \acrshort{mmse} of one-step prediction of $ \rv{X}_n $ from its infinite past, i.e.,
	\begin{equation}
		P_e(\rv{X}) = \inf_{\{b_i\}} \Exp{\left(\rv{X}_n - \sum_{i=1}^{\infty} b_i \rv{X}_{n-i}\right)}.
	\end{equation}
	The difference here is that we consider noisy prediction in \eqref{eq:prediction_error_finiteL}.
\end{remark}

The compression scheme in \autoref{figure:predictive_test_channel} is based on a corresponding scheme for standard Gaussian rate-distortion problem with \acrshort{mmse} distortion, developed in \cite{Zamir2008}. The main difference is that 
the problem considered here is essentially a remote source coding problem with log-loss. In particular, this affects the realization of the shaping filter and the derivation of mutual information fidelity. Our result is as follows:
\begin{theorem} \label{theorem:prediction_gaussian_ib}
	Let a bivariate Gaussian stationary source $ (\{\rv{X}_t , \rv{Y}_t\}) $ with \acrshort{snr} spectrum $ \Gamma(f) $ be given,  and assume that $ g(\cdot) $ achieves the optimum infinite order prediction error $ \sigma_{\infty} $. 
	If the quantizer in  \autoref{figure:predictive_test_channel} is replaced by an \acrshort{awgn} channel then the system of \autoref{figure:predictive_test_channel}, satisfies
	\begin{equation}
		\mathcal{R}(C) =  \frac{1}{2} \int_{-1/2}^{1/2}  \log \left[ \frac{1+\Gamma(f)}{1+D_{\theta}(f)}\right] \mathrm{d} f ,
	\end{equation}
	where $ D_{\theta}(f) $ is as defined in \eqref{eq:prediction_Dtheta_definition}, and $\theta$ is chosen such that
	\begin{equation}
		C =  \frac{1}{2} \int_{-1/2}^{1/2}  \log \left[ \frac{\Gamma(f)}{D_{\theta}(f)} \right] \mathrm{d} f.
	\end{equation}
\end{theorem}

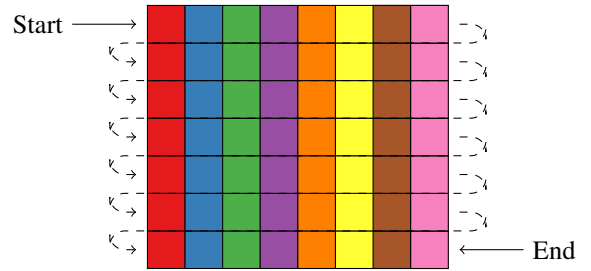
\begin{figure}[b]
    \centering
    \definecolor{mycolor1}{RGB}{228,26,28}
\definecolor{mycolor2}{RGB}{55,126,184}
\definecolor{mycolor3}{RGB}{77,175,74}
\definecolor{mycolor4}{RGB}{152,78,163}
\definecolor{mycolor5}{RGB}{255,127,0}
\definecolor{mycolor6}{RGB}{255,255,51}
\definecolor{mycolor7}{RGB}{166,86,40}
\definecolor{mycolor8}{RGB}{247,129,191}

\begin{tikzpicture}
    
    \draw[->] (-1,3.25) node[left] {Start} -- (-0.15,3.25);
    
    \foreach \y in {0,0.5,...,3}
    {
        \fill[mycolor1,draw = black] (0,0+\y) rectangle (0.5,0.5+\y);
        
        \fill[mycolor2,draw = black] (0.5,0+\y) rectangle (1,0.5+\y);
        
        \fill[mycolor3,draw = black] (1,0+\y) rectangle (1.5,0.5+\y);
        
        \fill[mycolor4,draw = black] (1.5,0+\y) rectangle (2,0.5+\y);
        
        \fill[mycolor5,draw = black] (2,0+\y) rectangle (2.5,0.5+\y);

        \fill[mycolor6,draw = black] (2.5,0+\y) rectangle (4,0.5+\y);

        \fill[mycolor7,draw = black] (3,0+\y) rectangle (3.5,0.5+\y);

        \fill[mycolor8,draw = black] (3.5,0+\y) rectangle (4,0.5+\y);
       }
       
      \draw[<-] (4.15,0.25)  -- (5,0.25) node[right] {End};
      
      \draw[->,dashed,rounded corners=2mm] (4.15,3.25) -| ++(0.35,-0.25) -| (-0.5,2.75) -- (-0.15,2.75);
      
      \draw[->,dashed,rounded corners=2mm] (4.15,2.75) -| ++(0.35,-0.25) -| (-0.5,2.25) -- (-0.15,2.25);
      
      \draw[->,dashed,rounded corners=2mm] (4.15,2.25) -| ++(0.35,-0.25) -| (-0.5,1.75) -- (-0.15,1.75);
      
      \draw[->,dashed,rounded corners=2mm] (4.15,1.75) -| ++(0.35,-0.25) -| (-0.5,1.25) -- (-0.15,1.25);
      
      \draw[->,dashed,rounded corners=2mm] (4.15,1.25) -| ++(0.35,-0.25) -| (-0.5,0.75) -- (-0.15,0.75);
      
      \draw[->,dashed,rounded corners=2mm] (4.15,0.75) -| ++(0.35,-0.25) -| (-0.5,0.25) -- (-0.15,0.25);
\end{tikzpicture}
    \caption{Source Interleaving Scheme}
    \label{fig:inteleaving}
\end{figure}

\autoref{theorem:prediction_gaussian_ib} assumes that the quantizer operation is replaced by Gaussian noise. The discrepancy between the ideal Gaussian noise setting and the practical quantizer setting, can be resolved using \acrfull{vq} and interleaving. Utilizing high-dimensional lattice code for quantization results in nearly Gaussian quantization noise \cite{zamir2014}. However, a direct application of a  \acrshort{vq} is not well-suited to the sequential prediction scheme of \autoref{figure:predictive_test_channel}. This problem can be resolved  by adding a spatial dimension to the problem, motivated by interleaving for \acrshort{isi} channels \cite{Guess2005}. Consider the scheme illustrated in \autoref{fig:inteleaving}. The input sequence to the encoder is rearranged in matrix form. The columns, whose elements are assumed to be memoryless, are used as the inputs to \acrshort{vq}. The rows on the other hand are fed to the prediction loop. Additional discussion on this technique can be found in \cite{Zamir2008}.


\begin{remark}
While the central block of the system is sequential and hence \textit{causal}, the whitening filter, shaping filter and the pre- and post filters are noncausal and therefore their realization requires delay. In particular, since $ h_{2,n} = h_{1,-n} $, if one of the filters is causal then the other must be anticausal. The problem is that usually the filter's response is infinite, therefore, the required delay is also unbounded. Indeed, the desired spectrum of the filters may be approximated to any order using filters of sufficiently large but finite delay $ \delta $. Then, \autoref{theorem:prediction_gaussian_ib} holds in general in the limit of infinite delay.
\end{remark}

	\section{Compound Information Bottleneck}
 	\label{section:comib_gaussian_processes}
	Following \cite{Dikshtein2022a} we next consider the \acrfull{comib} function of a bivariate stationary source with memory. Here, the \acrshort{comib} is given as a limit of normalized mutual information associated with vectors of source samples. For a real valued bivariate source $ (\{\rv{X}_t, \rv{Y}_t\}) $ (see \eqref{eq:bivariate_source}), normalized \acrshort{pf} constraint $ C_1 $, and normalized  \acrshort{ib} constraint $ C_2 $, the \acrshort{comib} rate can be written as
\begin{equation} \label{eq:oblib_rate_definition}
	\mathcal{R}^{comib}(C_1,C_2) = \lim_{ n \rightarrow \infty }  R_n^{comib}(C_1,C_2),
\end{equation}
where
\begin{equation} \label{eq:normalized_comib_function}
		R_n^{comib}(C_1,C_2) =
		\sup_{\pmf{\rv{Z}^n|\rv{Y}^n}} \inf_{\pmf{\rv{X}^n\rv{Y}^n}}
		\frac{1}{n} I(\rv{X}^n;\rv{Z}^n) 
\end{equation}
is the corresponding finite vector Gaussian \acrshort{comib} function. The optimization in \eqref{eq:normalized_comib_function} is over the sets $\{ \pmf{\rv{X}^n\rv{Y}^n} \} $ and $\{ \pmf{\rv{Z}^n|\rv{Y}^n} \}$, satisfying $\frac{1}{n} I(\rv{X}^n;\rv{Y}^n) \geq C_{1} $ and  $\frac{1}{n} I(\rv{Y}^n;\rv{Z}^n) \leq C_{2}$, respectively. 

A simple way to obtain solution to \eqref{eq:oblib_rate_definition} is by establishing a saddle point property. We briefly remind the reader this property as it will be used  in the proof. 
\begin{lemma}[Optimality of Saddle Point {\cite[Sec. 5.4.2]{Boyd2014}}] \label{lemma:optimality_saddle_point}
	Suppose there exists a saddle point $(\Tilde{w},\Tilde{z})$, satisfying
	$
	f(\Tilde{w},\Tilde{z}) = \inf_{w \in \set{W}} f(w,\Tilde{z}) 
	$ and $
	f(\Tilde{w},\Tilde{z}) = \sup_{z \in \set{Z}} f(\Tilde{w},z) 
	$,
	then
	\begin{equation}
		f(\Tilde{w},\Tilde{z}) = \sup_{z \in \set{Z}} \inf_{w \in \set{W}} f(w,z) .
	\end{equation}
\end{lemma}

Next, we consider the respective dual \acrshort{pf} problem \eqref{eq:normalized_pf_function} for Gaussian random variables. The problem in  \eqref{eq:pf_rate_definion} is rather delicate -- e.g., if $ (\rv{Y},\rv{Z}) $ are scalar jointly Gaussian random variables, the \acrshort{pf} rate is zero since one can use the channel from $ \rv{Y} $ to $ \rv{X} $ to describe the less significant bits of $ \rv{Y} $ \cite{Shamai2021}. Thus, additional constraints should be imposed here to have a non-trivial rate, as stated in the following theorem.
\begin{lemma} \label{lemma:pf_vector_gaussian}
	Suppose $\bv{X} \rightarrow \bv{Y} \rightarrow \bv{Z}$ constitute a jointly Gaussian vector Markov chain with positive definite marginal covariance matrices $\Sigma_{\bv{X}}$, $\Sigma_{\bv{Y}}$, and $\Sigma_{\bv{Z}}$ respectively, and the cross-covariance matrix of $\bv{Z}$ and $\bv{Y}$ is given by $\Sigma_{\bv{Z}\bv{Y}}$.
	Let $\Sigma_{\bv{Y}\bv{X}} $ be the cross-covariance matrix of the optimal solution to \eqref{eq:normalized_pf_function}.
	Further, let $U_1 \Phi V_1^T$ be the \acrfull{svd} of $\Sigma_{\bv{Y}}^{-1/2} \Sigma_{\bv{Y}\bv{X}} \Sigma_{\bv{X}}^{-1/2}$ and $U_2 \Psi V_2^T$ be the \acrshort{svd} 
	of $\Sigma_{\bv{Z}}^{-1/2} \Sigma_{\bv{Z}\bv{Y}} \Sigma_{\bv{Y}}^{-1/2}$.
	Then, the underlying Gaussian \acrshort{pf} problem \eqref{eq:normalized_pf_function} can be relaxed to the following optimization problem:
	\begin{equation} \label{eq:pf_rate_general_gaussian}
		\begin{aligned}
			&R_n^{\mathsf{PF}}(C_1) = \mkern-20mu
			&& \min_{\scaleto{U_1 \in \set{U}(n), \{\phi_{i}\}}{6pt}} 
			&& \mkern-20mu \minus \frac{1}{2 n} \log \det( I
			\minus U_1^T \Phi^2 U_1 V_2^T \Psi^2 V_2 
			) \\
			&&& \text{subject to} 
			&& \mkern-10mu  -\frac{1}{2n} \sum_{i=1}^n   \log(1-\phi_i^2) = C_1,
		\end{aligned}
	\end{equation}
	where $\set{U}(n)$ is the set of all $n \times n$ unitary matrices (the unitary group), and $\{\phi_i\}$ are the entries of the diagonal matrix $\Phi$.
\end{lemma}

Note that in contrast to \autoref{lemma:ib_jointly_gaussian_vectors}, we do not have a closed form-solution for the general Gaussian \acrshort{pf} problem, and additional numerical optimization is needed. However, for the compound setting we obtain an exact solution which incorporates a white \acrshort{snr} spectrum. Such spectrum has a constant intensity over the entire bandwidth of the signal.

\begin{theorem} \label{theorem:comib_discrete_gaussian_RP}
    Suppose that \eqref{eq:oblib_rate_definition} is evaluated for the scenario where $\{\rv{X}_t \}$ and $ \{\rv{Y}_t \}$ are jointly Gaussian random processes with marginal power spectral densities $S_{\rv{X}}(f)$ and $S_{\rv{Y}}(f)$ respectively. 
    The resulting optimal channel from $ \rv{X}_n $ to $ \rv{Z}_n $ has a linear form, i.e.,
    \begin{align}
        \rv{Y}_n &= h_n * \rv{X}_n + \rv{W}_n, \\
    	\rv{Z}_n &= g_n * \rv{Y}_n + \rv{V}_n,
    \end{align}
where:
	  $ h_n $ and $ g_n $ are impulse responses of liner time-invariant filters, whose absolute squared value frequency responses are given by:
	\begin{align}
		|H(f)|^2 &= \frac{1}{1+\gamma} \frac{S_{\rv{Y}}(f)}{S_{\rv{X}}(f)}, \quad
		|G(f)|^2 = \frac{1}{1+\lambda} , \\
		S_{\rv{W}}(f) &= \frac{1}{1+\gamma}  \frac{S_{\rv{Y}}(f)}{S_{\rv{X}}(f)}, \quad
		S_{\rv{V}}(f) = \frac{1}{1+\lambda} ;
	\end{align}
	$ \gamma = 2^{2C_1}-1 $; and $ \lambda = 2^{2C_2}-1 $.
Therefore, the optimal double-sided \acrshort{snr} spectrum is white.
The corresponding oblivious rate is
\begin{equation}
    \mathcal{R}^{comib}(C_1,C_2) = -\frac{1}{2} \log [1-(1-2^{-2C_1})(1-2^{-2C_2})].
\end{equation}
\end{theorem}

This theorem is obtained by applying a saddle-point property from \autoref{lemma:optimality_saddle_point} on \autoref{lemma:ib_jointly_gaussian_vectors} and \autoref{lemma:pf_vector_gaussian} and has the following practical implication:  The most robust approach in case there is no information regarding the structure of the observed signal and noise is to assume the input is white.

	\section{Concluding Remarks}
	In this paper, we have addressed the jointly Gaussian process \acrshort{ib} problem. A water-filling type solution has been obtained. Then, a linear prediction scheme that attains the \acrshort{ib} was proposed and analyzed. Finally, a closed-form solution has been given to a compound version of the  \acrshort{ib} for Gaussian processes.

Future research directly related to the results of this paper calls for further investigating single-letter quantization algorithms with information-theoretic metrics. In addition, it would be interesting to consider the \acrshort{ib} and \acrshort{pf} problems, when the constraint is not $I(\rv{Y};\rv{Z}) \leq C$, but $H(\rv{Z}) \leq C $, as also done in \cite{Homri2018}. Then the transformation $\rv{Y} \rightarrow \rv{Z}$ is known to be discrete.

	\section*{Acknowledgment}
	The work has been supported by the 
	European Union's Horizon 2020 Research And Innovation Programme,
	grant agreement no. 694630, by the ISF under Grant 1791/17, and by the WIN consortium via the
	Israel minister of economy and science.

\ifthenelse{\boolean{fullver}}{
\newpage
\allowdisplaybreaks
\appendix

	\subsection{Information Bottleneck for jointly Gaussian $ (\rv{X},\rv{Y}) $}
	The following auxiliary result is well known and mentioned here for self-sustainability. For alternative variant of this result, the interested reader is referred to \cite{Guo2013}. The uniqueness of the result and the proof provided here is the specific application of the \acrshort{epi} on the \acrshort{ib} problem (rather than I-MMSE on \acrshort{infcomb} as in \cite{Guo2013}).
\begin{lemma} \label{lemma:comib_scalar_gaussian_ib}
    Suppose that $\rv{X} \rightarrow \rv{Y} \rightarrow \rv{Z}$ constitute a Markov chain, where $\rv{X}$ and $\rv{Y}$ are unit variance jointly Gaussian random variables with correlation $\rho_1$. 
    Then, it holds that the value of the \acrshort{ib} program is \cite{Tishby1999}
    \begin{align}
        R^{IB} (C_2) &=
         \max_{\pmf{\bv{Z}|\bv{Y}} \colon I(\rv{Y};\rv{Z}) = C_2}
         I(\rv{X};\rv{Z}) \\
        &=\frac{1}{2} \log \frac{1}{1-\rho_1^2 \rho_2^2},
    \end{align}
    where $\rho_2^2 = 1-2^{-2C_2}$,
    and the optimizing distribution is a jointly Gaussian triplet $(\rv{X},\rv{Y},\rv{Z})$ with covariance matrix
    \begin{equation}
        \begin{pmatrix}
            1 & \rho_1 & \rho_1 \rho_2 \\
            \rho_1 & 1 & \rho_2 \\
            \rho_1 \rho_2 & \rho_2 & 1
        \end{pmatrix}.
    \end{equation}
\end{lemma}
\begin{proof}
The main tools used in  proof of the lemma are the \acrfull{epi} \cite{Dembo1991}, and the scaling property of the differential entropy function \cite{Cover2006}. For any two zero mean unit variance jointly Gaussian random variables $ \rv{X} $ and $ \rv{Y} $, with correlation $ \rho_1 $, there exists a unit variance Gaussian random variable $\rv{W}$, independent of $\rv{Y}$, i.e.,
 \begin{equation}
 	\rv{X} = \rho_1 \rv{Y} + \sqrt{1-\rho_1^2} \cdot \rv{W}.
  \end{equation}
 The objective function has the following form in the Gaussian case:
\begin{align}
    I(\rv{X};\rv{Z})
    &= h(\rv{X}) - h(\rv{X}|\rv{Z}) \\
    &= \frac{1}{2} \log 2\pi e - h \left( \rho_1 \rv{Y} + \sqrt{1-\rho_1^2} \cdot  \rv{W} \bigg| \rv{Z} \right).
\end{align}
The last term can be bounded from below using \acrshort{epi},
\begin{align}
    &h \left( \rho_1 \rv{Y} + \sqrt{1-\rho_1^2} \cdot \rv{W} \bigg| \rv{Z} \right) \\
    &\geq \frac{1}{2} \log \left( 2^{2 h(\rho_1 \rv{Y}  | \rv{Z})} + 2^{2 h\left(\sqrt{1-\rho_1^2} \cdot \rv{W}  | \rv{Z}\right)} \right) \\
    &= \frac{1}{2} \log \left( \rho_1^2 \cdot 2^{2 h( \rv{Y}  | \rv{Z})} +(1-\rho_1^2) \cdot 2^{2 h( \rv{W}  )} \right) \\
    &\eqann{=}{a} \frac{1}{2} \log \left( \rho_1^2 \cdot 2^{2(h(\rv{Y}) -  I( \rv{Y}  ; \rv{Z}))} +(1-\rho_1^2) 2\pi e \right) \\
    &\eqann{=}{b} \frac{1}{2} \log 2\pi e \left( \rho_1^2 2^{- 2C_2} + 1-\rho_1^2 \right),
\end{align}
where \eqannref{a} and \eqannref{b} follow since $h(\rv{Y}) = h(\rv{W}) =  \frac{1}{2} \log 2\pi e$.
Thus,
\begin{equation} \label{eq:ib_gaussian_ub}
    I(\rv{X};\rv{Z}) \leq \frac{1}{2} \log \frac{1}{1- \rho_1^2(1-2^{-2C_2})}.
\end{equation}
Since the inequality in \eqref{eq:ib_gaussian_ub} follows from \acrshort{epi}, it can be attained with equality if we choose $\rv{Z} \sim \mathcal{N}(0,1)$, $\rv{V} \sim \mathcal{N}(0,1)$ and
\begin{equation}
    \rv{Y} = \rho_2 \rv{Z} + \sqrt{1-\rho_2^2} \rv{V},
\end{equation}
such that
\begin{equation}
    \rv{X} = \rho_1 \rho_2 \rv{Z} + \rho_1 \sqrt{1-\rho_2^2} \rv{V} + \sqrt{1-\rho_1^2} \rv{W},
\end{equation}
with
    $\rho_2^2 = 1-2^{-2C_2}$. This completes the proof of the lemma.
\end{proof}
	
	\subsection{Equivalence  of Covariance Measures}
	Let $ (\bv{X},\bv{Y}) $ be a pair of jointly Gaussian random vectors with marginal covariance matrices $ \Sigma_{\bv{X}} $, $ \Sigma_{\bv{Y}} $, and cross-covariance matrix $ \Sigma_{\bv{X}\bv{Y}} $. Without loss of generality, they can be represented using a linear form as
\begin{equation}
	\bv{Y} = K_{yx} \bv{X} + \bv{W},
\end{equation}
where $ K_{yx} = \Sigma_{\bv{X}\bv{Y}}^T \Sigma_{\bv{X}}^{-1} $, and $ \Sigma_{\bv{W}} = \Sigma_{\bv{Y}} - \Sigma_{\bv{X}\bv{Y}}^T \Sigma_{\bv{X}}^{-1} \Sigma_{\bv{X}\bv{Y}} $.  

We define and compare two linear measures:
\begin{itemize}
	\item Normalized Correlation Matrix:
	\begin{equation}
		\Sigma_{cor} \triangleq \Sigma_{\bv{X}}^{-\frac{1}{2}}  \Sigma_{\bv{X}\bv{Y}} \Sigma_{\bv{Y}}^{-\frac{1}{2}}.
	\end{equation}
	\item \acrshort{snr} Matrix:
	\begin{equation}
		\Sigma_{snr} \triangleq \Sigma_{\bv{W}}^{-\frac{1}{2}} K_{yx} \Sigma_{\bv{X}} K_{yx}^T \Sigma_{\bv{W}}^{-\frac{1}{2}}.
	\end{equation}
\end{itemize}

By the \acrfull{svd} theorem \cite{Horn2012}, there exist unitary matrices $U$ and $V$ and a diagonal matrix $D$ such that 
\begin{equation}
	\Sigma_{\bv{X}}^{-\frac{1}{2}}  \Sigma_{\bv{X}\bv{Y}} \Sigma_{\bv{Y}}^{-\frac{1}{2}}  = U D V^T.
\end{equation}
Similarly, there exist unitary matrix $ O $ and diagonal matrix $ \Gamma $ such that
\begin{equation}
	\Sigma_{snr} = O \Gamma O^T.
\end{equation}
\begin{lemma} \label{lemma:equivalence_of_measures}
	There is a one-to-one correspondence between the elements of $ D $ and $ \Gamma $, i.e.,
	\begin{equation}
		\gamma_i \leftrightarrow \frac{d_i^2}{1-d_i^2}.
	\end{equation}
\end{lemma}
\begin{proof}
	Note that
	\begin{align}
		&\Sigma_{\bv{W}} 
		= \Sigma_{\bv{Y}} - \Sigma_{\bv{X}\bv{Y}}^T \Sigma_{\bv{X}}^{- 1} \Sigma_{\bv{X}\bv{Y}} \\
		&= \Sigma_{\bv{Y}}^{\frac{1}{2}}(I \minus \Sigma_{\bv{Y}}^{\minus \frac{1}{2}} \Sigma_{\bv{X}\bv{Y}}^T \Sigma_{\bv{X}}^{\minus \frac{1}{2}} \Sigma_{\bv{X}}^{\minus \frac{1}{2}} \Sigma_{\bv{X}\bv{Y}} \Sigma_{\bv{Y}}^{\minus \frac{1}{2}}) \Sigma_{\bv{Y}}^{\frac{1}{2}} \\
		&= \Sigma_{\bv{Y}}^{\frac{1}{2}}(I - V D U^T U D V^T) \Sigma_{\bv{Y}}^{\frac{1}{2}} \\
		&= \Sigma_{\bv{Y}}^{\frac{1}{2}} V (I -  D^2 )  V^T \Sigma_{\bv{Y}}^{\frac{1}{2}}.
	\end{align}
	Thus,
	\begin{equation}
		\Sigma_{\bv{W}} ^{-1} = \Sigma_{\bv{Y}}^{-\frac{1}{2}} V (I -  D^2 )^{-1}  V^T \Sigma_{\bv{Y}}^{-\frac{1}{2}}.
	\end{equation}
	Also note that
	\begin{align}
		&K_{yx} \Sigma_{\bv{X}} K_{yx}^T
		= \Sigma_{\bv{X}\bv{Y}}^T \Sigma_{\bv{X}}^{-1} \Sigma_{\bv{X}\bv{Y}} \\
		&= \Sigma_{\bv{Y}}^{\frac{1}{2}} \Sigma_{\bv{Y}}^{-\frac{1}{2}}\Sigma_{\bv{X}\bv{Y}}^T \Sigma_{\bv{X}}^{-\frac{1}{2}}\Sigma_{\bv{X}}^{-\frac{1}{2}} \Sigma_{\bv{X}\bv{Y}} \Sigma_{\bv{Y}}^{-\frac{1}{2}} \Sigma_{\bv{Y}}^{\frac{1}{2}} \\
		&= \Sigma_{\bv{Y}}^{\frac{1}{2}} V D U^T U D V^T \Sigma_{\bv{Y}}^{\frac{1}{2}} \\
		&= \Sigma_{\bv{Y}}^{\frac{1}{2}} V D^2  V^T \Sigma_{\bv{Y}}^{\frac{1}{2}} .
	\end{align}
	Finally, consider the determinant of the \acrshort{svd} form of the \acrshort{snr} matrix. We have
	\begin{align}
		\prod_{i=1}^{n} \gamma_i 
		&= |\Gamma| \\
		&= | O O^T| \cdot | \Gamma | \\
		&= | O  \Gamma  O^T | \\
		&= | \Sigma_{\bv{W}}^{-\frac{1}{2}} K_{yx} \Sigma_{\bv{X}} K_{yx}^T \Sigma_{\bv{W}}^{-\frac{1}{2}}| \\
		&= | \Sigma_{\bv{W}}^{-1} | \cdot | K_{yx} \Sigma_{\bv{X}} K_{yx}^T | \\
		&= |I - D^2|^{-1} \cdot | \Sigma_{\bv{Y}}^{-1}|  \cdot | \Sigma_{\bv{Y}}| |D^2| \\
		&= \prod_{i=1}^n \frac{d_i^2}{1-d_i^2}.
	\end{align}
	This implies that there exists a one-to-one and onto mapping between $\{\gamma_i\}$ and $d_i$.
\end{proof}

	\subsection{Proof of  \autoref{lemma:ib_jointly_gaussian_vectors}}
Suppose $\bv{X}$ and $\bv{Y}$ are jointly Gaussian  random vectors with covariance matrix $\Sigma_{\bv{X}\bv{Y}}$. It is easy to verify that without loss of generality, we can write
\begin{equation}
    \bv{X} = K_{xy} \bv{Y} + \bv{W},
\end{equation}
where $K_{xy} = \Sigma_{\bv{X}\bv{Y}} \Sigma_{\bv{Y}}^{-1} $  and $\Sigma_{\bv{W}} = \Sigma_{\bv{X}} - \Sigma_{\bv{X}\bv{Y}} \Sigma_{\bv{Y}}^{-1} \Sigma_{\bv{X}\bv{Y}}^T$.

By the \acrfull{svd} theorem \cite{Horn2012}, there exist unitary matrices $U$ and $V$ and a diagonal matrix $D$ such that 
\begin{equation}
    \Sigma_{\bv{X}}^{-\frac{1}{2}}  \Sigma_{\bv{X}\bv{Y}} \Sigma_{\bv{Y}}^{-\frac{1}{2}}  = U D V^T.
\end{equation}
Denote $T_x \triangleq U^T \Sigma_{\bv{X}}^{-\frac{1}{2}}$,  and let $\tilde{\bv{X}} \triangleq T_x \bv{X}$.  It follows that $ \Sigma_{\tilde{\bv{X}}}  = I$,
implying $\tilde{\bv{X}}$ is a random Gaussian vector with independent unit variance entries. Similarly defining $\tilde{\bv{Y}} = T_y \bv{Y}$ where $T_y = V^T \Sigma_{\bv{Y}}^{-\frac{1}{2}}$, results in $\tilde{\bv{Y}}$ being a random Gaussian vector with unit-variance independent entries, i.e, $\Sigma_{\tilde{\bv{Y}}}  = I$. 
Moreover, note that,
\begin{align}
    \tilde{\bv{X}} 
    &= T_x \bv{X} \\
    &= U^T \Sigma_{\bv{X}}^{-\frac{1}{2}} \left( K_{xy} \bv{Y} + \bv{W} \right) \\
    &= U^T \Sigma_{\bv{X}}^{-\frac{1}{2}} \Sigma_{\bv{X}\bv{Y}} \Sigma_{\bv{Y}}^{-\frac{1}{2}} V V^T \Sigma_{\bv{Y}}^{-\frac{1}{2}} \bv{Y} + U^T \Sigma_{\bv{X}}^{-\frac{1}{2}} \bv{W}  \\
    &= U^T U D V^T V \tilde{\bv{Y}} + U^T \Sigma_{\bv{X}}^{-\frac{1}{2}} \bv{W}  \\
    &= D \tilde{\bv{Y}} + U^T \Sigma_{\bv{X}}^{-\frac{1}{2}} \bv{W}.
\end{align}
Defining $\tilde{\bv{W}} \triangleq U^T \Sigma_{\bv{X}}^{-\frac{1}{2}} \bv{W}$, we obtain
\begin{align}
    \Sigma_{\tilde{\bv{W}}}
    &= \Exp{\tilde{\bv{W}}\tilde{\bv{W}}^T}\\
    &= U^T \Sigma_{\bv{X}}^{-\frac{1}{2}} \Sigma_{\bv{W}} \Sigma_{\bv{X}}^{-\frac{1}{2}} U \\
    &= U^T \Sigma_{\bv{X}}^{-\frac{1}{2}} \left( \Sigma_{\bv{X}} - \Sigma_{\bv{X}\bv{Y}} \Sigma_{\bv{Y}}^{-1} \Sigma_{\bv{X}\bv{Y}}^T
    \right)
    \Sigma_{\bv{X}}^{-\frac{1}{2}} U \\
    &= I - D^2.
\end{align}
Thus, the transformed channel is a parallel channel, satisfying
\begin{equation} \label{eq:ib_vector_gaussian_proof_parallel_channel}
    \tilde{\rv{X}}_i = d_i \tilde{\rv{Y}}_i +\sqrt{1-d_i^2} \tilde{\rv{W}}_i.
\end{equation}

Consider the mutual information constraint on the pair $(\bv{Y},\rv{Z})$. Since the linear transform $T_y$ has full rank, $\bv{Y}$ and $\tilde{\bv{Y}}$ are a bijection, and there is no loss of information, i.e.,
\begin{align}
    I(\bv{Y};\rv{Z})
    &= I(\tilde{\bv{Y}};\rv{Z}) \\
    &= h(\tilde{\bv{Y}}) - h(\tilde{\bv{Y}}|\rv{Z}) \\
    &= \sum_{i=1}^N h(\tilde{\rv{Y}}_i) - h(\tilde{\rv{Y}}_i|\tilde{\rv{Y}}^{i-1}, \rv{Z}).
\end{align}
Identifying $\rv{Z}_i \triangleq (\rv{Z},\tilde{\rv{Y}}^{i-1})$ and denoting
$C_{i} = I(\tilde{\rv{Y}}_i; \rv{Z}_i) $ we obtain the following representation of the bottleneck constraint
\begin{equation}
    \sum_{i=1}^{n_y} C_{i}  \leq C, \qquad C_{i} = I(\tilde{\rv{Y}}_i; \rv{Z}_i).
\end{equation}

Now consider the objective function. Similarly, since $ T_x $ is invertible, the linear mapping from $\tilde{\bv{X}}$ to $\bv{X}$ is a bijection, there is no loss of information,
\begin{align}
    I(\bv{X};\rv{Z})
    &= I(\tilde{\bv{X}};\rv{Z}) \\
    &= h(\tilde{\bv{X}}) - h(\tilde{\bv{X}}|\rv{Z}) \\
    &= \sum_{i=1}^n h(\tilde{\rv{X}}_i) - h(\tilde{\rv{X}}_i|\tilde{\rv{X}}^{i-1}, \rv{Z})\\
    &\eqann{\leq}{a} \sum_{i=1}^n h(\tilde{\rv{X}}_i) - h(\tilde{\rv{X}}_i|\tilde{\rv{X}}^{i-1},\tilde{\rv{Y}}^{i-1}, \rv{Z}) \\
    &\eqann{=}{b} \sum_{i=1}^n h(\tilde{\rv{X}}_i) - h(\tilde{\rv{X}}_i|\tilde{\rv{Y}}^{i-1}, \rv{Z}) \\
    &= \sum_{i=1}^n h(\tilde{\rv{X}}_i) - h(\tilde{\rv{X}}_i| \rv{Z}_i) \\
    &= \sum_{i=1}^n I(\tilde{\rv{X}}_i;\rv{Z}_i),
\end{align}
where \eqannref{a} follows since conditioning reduces differential entropy, and equality in \eqannref{b} is due to Markov chain $\tilde{\rv{X}}_i \rightarrow (\tilde{\rv{Y}}^{i-1},\rv{Z}) \rightarrow \tilde{\rv{X}}^{i-1}$.
Further, by \autoref{lemma:comib_scalar_gaussian_ib}, and the parallel channel construction from \eqref{eq:ib_vector_gaussian_proof_parallel_channel}, we have
\begin{equation}
     I(\tilde{\rv{X}}_i;\rv{Z}_i)
     \leq \frac{1}{2} \log \frac{1}{1-d_i^2 (1-2^{-2C_{i}})},
\end{equation}
with equality achieved for $\rv{Z}_i \sim \mathcal{N}(0,1) $, $\rv{V}_i \sim \mathcal{N}(0,1)$ and
\begin{equation}
    \tilde{\rv{Y}}_i = \rho \rv{Z}_i + \sqrt{1-\rho^2} \rv{V}_i,
\end{equation}
where $\rho^2 = 1-2^{-2C_{i}}$.
Thus, we have relaxed our original optimization problem in \eqref{eq:normalized_ib_function} to the following one:
\begin{equation}
    \begin{aligned}
    &R_n^{ib}(C) =
    && \underset{\{C_{i} \}}{\text{maximize}}
    && \frac{1}{2 n} \sum_{i=1}^n  \log \frac{1}{1-d_i^2 (1-2^{-2C_{i}})} \\
    &&& \text{subject to}
    &&  \frac{1}{n}\sum_{i=1}^n C_{i} \leq C.
    \end{aligned}
\end{equation}
We apply KKT conditions to solve the underlying optimization problem. The respective Lagrangian is given by
\begin{multline}
    L(\{C_{i}\},\lambda) - \lambda C \\
    = \frac{1}{n}\sum_{i=1}^n \frac{1}{2} \log \frac{1}{1-d_i^2 (1-e^{-2C_{i}})} -  \frac{\lambda}{n}
    \sum_{i=1}^n C_{i} 
    .
\end{multline}
The KKT conditions are given by:
\begin{itemize}
    \item Stationarity:
    \begin{equation}
        \frac{\partial L}{\partial C_{i}} 
        = \frac{1}{n} \cdot  \frac{d_i^2 e^{-2C_{i}}}
        {1-d_i^2(1-e^{-2C_{i}})} -\frac{\lambda}{n} = 0.
    \end{equation}
    Thus,
    \begin{equation}
        (1-d_i^2) e^{2C_{i}} +d_i^2 = \frac{d_i^2}{\lambda} \rightarrow C_{i} = \frac{1}{2} \ln \frac{d_i^2 (1 -\lambda) }{ \lambda(1-d_i^2)}.
    \end{equation}
    \item Complementary Slackness:
    \begin{equation}
        \lambda \left(
    \frac{1}{n} \sum_{i=1}^n C_{i} - C
    \right) = 0.
    \end{equation}
    \item Constraints:
    \begin{align}
        C_{i} &\geq 0, \\
        \frac{1}{n}\sum_{i=1}^{n} C_{i} &\leq C.
        \label{eq:sumC_connnstraint}
    \end{align}
\end{itemize}
Since $\lambda = 0$ is infeasible solution, therefore $\lambda^* > 0 $ and the last constraint must be satisfied with equality. Thus $\lambda$ is chosen as the solution to
\begin{equation}
   \frac{1}{n} \sum_{i=1}^n C_{i} = C.
\end{equation}
Further denote $\nu = \frac{1}{\lambda}$, the optimal solution has the following water-filling form:
\begin{equation}
    C_{i} =  \frac{1}{2}\log\left[ \frac{d_i^2 (\nu -1)}{1-d_i^2} \right]^+ .
\end{equation}
Thus,
\begin{equation}
    R_n^{ib}(C) =  -\frac{1}{2 n} 
    \sum_{i=1}^n
    \log 
    \left[1-d_i^2\left(1 -
    \left[ \frac{1-d_i^2}{d_i^2 (\nu -1)} \right]^-
    \right)
    \right],
\end{equation}
where $\nu$ is chosen to satisfy the rate constraint \eqref{eq:sumC_connnstraint} with equality.

Thus, the maximum is achieved by a jointly Gaussian triple $(\bv{X},\bv{Y},\rv{Z})$ and
\begin{equation}
	R_n^{ib}(C) 
	=
	\frac{1}{2n} \sum_{i=1}^n
	\log
	\left[
	\frac{\nu-1}{\nu (1-d_i^2)}
	\right]^+,
\end{equation}
where the water-filling level $\nu$ is chosen such that
\begin{equation}
	\frac{1}{2n}\sum_{i=1}^n  \log\left[ \frac{d_i^2 (\nu -1)}{1-d_i^2} \right]^+ = C.
\end{equation}

Note that $\{d_i\}$ are eigenvalues of $ \Sigma_{\bv{X}}^{-\frac{1}{2}}  \Sigma_{\bv{X}\bv{Y}} \Sigma_{\bv{Y}}^{-\frac{1}{2}}  $, which is a normalized correlation type measure. We would like to further represent the eigenvalues with a different intuitive form.
We start by writing the channel from $\bv{X} $ to $\bv{Y}$ in the following form:
\begin{equation}
	\bv{Y} = K_{yx} \bv{X} + \tilde{\bv{W}},
\end{equation}
where $K_{yx} = \Sigma_{\bv{X}\bv{Y}}^T \Sigma_{\bv{X}}^{-1} $  and $\Sigma_{\tilde{\bv{W}}} = \Sigma_{\bv{Y}} - \Sigma_{\bv{X}\bv{Y}}^T \Sigma_{\bv{X}}^{-1} \Sigma_{\bv{X}\bv{Y}}$.  Let $O \Gamma O^T$ be the \acrshort{svd} of the corresponding signal to noise matrix defined by
\begin{equation}
	O \Gamma O^T \triangleq \Sigma_{\tilde{\bv{W}}}^{-1/2} K_{yx} \Sigma_{\bv{X}} K_{yx}^T \Sigma_{\tilde{\bv{W}}}^{-1/2},
\end{equation}
where $ \Gamma = \diag \{\gamma_i\} $.
We have already shown in \autoref{lemma:equivalence_of_measures}, that a bijection holds between $ \gamma_i $ and $ \frac{d_i^2}{1-d_i^2} $, thus, the optimal vector Gaussian \acrshort{ib} can be expressed in the following elegant form:
\begin{equation}
	R_n^{ib}(C) 
	=
	\frac{1}{2n} \sum_{i=1}^n
	\log
	\left[
	\frac{\nu-1}{\nu } (1+\gamma_i)
	\right]^+,
\end{equation}
where the water-filling level $\nu$ is chosen such that
\begin{equation}
	\frac{1}{2n}\sum_{i=1}^n \log\left[ (\nu -1) \gamma_i \right]^+ = C.
\end{equation}
Further, by  defining 
\begin{equation}
	\theta \triangleq \frac{1}{\nu-1}
\end{equation}
we obtain the desired result.


	\subsection{Mutual Information Rate for Gaussian Random Processes}
	Let $ \{\rv{X}_t\}_{t=-\infty}^{+\infty} $ and $ \{\rv{Y}_t\}_{t=-\infty}^{+\infty} $ be a pair of discrete-time stationary jointly Gaussian stochastic processes with auto-correlation and cross-correlation functions
\begin{align}
	R_{\rv{X}}[k]  &\triangleq \Exp{\rv{X}_t \rv{X}_{t+k}}, \\
	R_{\rv{Y}}[k]  &\triangleq \Exp{\rv{Y}_t \rv{Y}_{t+k}}, \\
	R_{\rv{X} \rv{Y}}[k]  &\triangleq \Exp{\rv{X}_t \rv{Y}_{t+k}}, \\
	R_{\rv{Y}\rv{X}}[k]  &\triangleq \Exp{\rv{Y}_t \rv{X}_{t+k}} = R_{\rv{X}\rv{Y}}[-k].
\end{align}
The respective power spectral densities and cross-power spectral densities exist via the Fourier transform
\begin{align}
	S_{\rv{X}}(f) &= \sum_{k=- \infty}^{k = \infty} R_{\rv{X}}[k] e^{-j2\pi k f}, \\
	S_{\rv{Y}}(f) &= \sum_{k=- \infty}^{k = \infty} R_{\rv{Y}}[k] e^{-j2\pi k f}, \\
	S_{\rv{X}\rv{Y}}(f) &= \sum_{k=- \infty}^{k = \infty} R_{\rv{X} \rv{Y}}[k] e^{-j2\pi k f}, \\
	S_{\rv{Y} \rv{X}}(f) &= \sum_{k=- \infty}^{k = \infty} R_{\rv{Y}\rv{X}}[k] e^{-j2\pi k f} = S_{\rv{X}\rv{Y}}^*(f).
\end{align}
The corresponding \acrshort{snr} spectrum is defined by
\begin{equation}
	\Gamma(f )  \triangleq \frac{S_{\rv{X} \rv{Y}}(f) }{S_{\rv{X} }(f) S_{\rv{Y} }(f)  - S_{\rv{X} \rv{Y}}(f) }.
\end{equation}
The mutual-information rate is defined by  \cite{Komaee2020}
\begin{equation}
	\mathcal{I}(\{\rv{X}_t\}; \{\rv{Y}_t\}) = \lim_{n \rightarrow \infty} \frac{1}{n} I(\rv{X}_1,\dots, \rv{X}_n;\rv{Y}_1,\dots, \rv{Y}_n).
\end{equation}
\begin{lemma} \label{lemma:mutual_information_rate_gaussian_processes}
	The mutual information rate between two jointly Gaussian random processes with \acrshort{snr} spectrum $ \Gamma(f) $ is given by
	\begin{equation}
		\mathcal{I}(\{\rv{X}_t\}; \{\rv{Y}_t\}) = \frac{1}{2} \int_{-1/2}^{1/2} \log (1+\Gamma(f)) \mathrm{d} f.
	\end{equation}
\end{lemma}
\begin{proof}
Fix $ n $ and denote $ \bv{X} = \rv{X}_1^n $, $ \bv{Y} = \rv{Y}_1^n $.
Since $ \bv{X} $ and $ \bv{Y} $ are jointly Gaussian random vectors, they can be represented using a linear form,
\begin{equation}
	\bv{X} = H \bv{Y} + \bv{W},
\end{equation}
where $ H = \Sigma_{\bv{X}\bv{Y}} \Sigma_{\bv{Y}}^{-1}  $, and  $ \bv{W} \sim \mathcal{N}(0,\Sigma_{\bv{X}} - \Sigma_{\bv{X}\bv{Y}} \Sigma_{\bv{Y}}^{-1} \Sigma_{\bv{X}\bv{Y}}^T) $.
It follows that
\begin{align}
	I(\bv{X};\bv{Y}) 
	&= h(\bv{X}) - h(\bv{X}|\bv{Y}) \\
	&= h(H\bv{Y}+\bv{W}) - h(\bv{W}) \\
	&= \frac{1}{2} \log \frac{|H \Sigma_{\bv{Y}} H^T + \Sigma_{\bv{W}}|}{|\Sigma_{\bv{W}}|} \\
	&= \frac{1}{2} \log  |\Sigma_{\bv{W}}^{-1/2} H \Sigma_{\bv{Y}} H^T \Sigma_{\bv{W}}^{-1/2} + I|.
\end{align}
The \acrshort{svd} of the \acrshort{snr} matrix defined by $ \Sigma_{snr} \triangleq \Sigma_{\bv{W}}^{-1/2} H \Sigma_{\bv{Y}} H^T \Sigma_{\bv{W}}^{-1/2} $, is denoted by $ O \Gamma O^T $, where $ \Gamma = \diag \{\gamma_i\} $. Thus,
\begin{equation}
	I(\rv{X}^n;\rv{Y}^n) = I(\bv{X};\bv{Y}) = \frac{1}{2}  \sum_{i=1}^n  \log (1+\gamma_i).
\end{equation}
Applying Szeg\"o theorem \cite{Gray2006}, we obtain that the mutual-information rate is given by
\begin{align}
	\mathcal{I}(\{\rv{X}_t\};\{\rv{Y}_t\}) 
	&= \lim_{ n \rightarrow \infty } \frac{1}{n} I(\rv{X}^n;\rv{Y}^n) \\
	&= \lim_{ n \rightarrow \infty } \frac{1}{n} \frac{1}{2}  \sum_{i=1}^n  \log (1+\gamma_i) \\
	&= \frac{1}{2} \int_{-1/2}^{1/2} \log (1+\Gamma(f)) \mathrm{d} f.
\end{align}
\end{proof}
	
	\subsection{Proof of \acrshort{ib} for Gaussian Process \autoref{theorem:ib_gaussian_process}}
%
	Combining \autoref{lemma:ib_jointly_gaussian_vectors}, with \eqref{eq:ib_rate_definion} and Szeg\"o theorem \cite{Gray2006}, we have
	\begin{align}
		\mathcal{R}^{ib}(C) &= R_n^{ib}(C) \\
		&= \lim_{ n \rightarrow \infty } \frac{1}{2n}  \sum_{i=1}^n 
		\log
		\left[
		\frac{1+\gamma_i}{1+\theta}
		\right]^+ \\
		&= \frac{1}{2} \int_{-1/2}^{1/2}  \log \left[ \frac{1+\Gamma(f)}{1+\theta}\right]^+ \mathrm{d} f,
	\end{align}
	The reverse \textit{water} level  $ \theta $ is chosen such that the total rate is $ C $,
	\begin{equation}
		C = \lim_{ n \rightarrow \infty } \frac{1}{2n} \sum_{i=1}^n  \log\left[ \frac{ \gamma_i }{\theta}\right]^+  = \frac{1}{2} \int_{-1/2}^{1/2}  \log \left[ \frac{\Gamma(f)}{\theta} \right]^+ \mathrm{d} f,
	\end{equation}
	where $ \Gamma(f) $ is the \acrshort{snr} spectrum.
	
	
	\subsection{Proof of Gaussian Process Prediction from \autoref{theorem:prediction_gaussian_ib}}
Let $ g(\rv{V}_{i-L}^{i-1}) $ be the optimal \acrshort{mmse} predictor of $ \rv{U}_i $ based on $ \rv{V}_{i-L}^{i-1} $, then
\begin{align}
	&I(\{\rv{U}_t \} ;\rv{V}_i | \rv{V}_{i-L}^{i-1})  \\
	&= I(\{\rv{U}_t \},\rv{U}_i ;\rv{V}_i | \rv{V}_{i-L}^{i-1}) \\
	&= I(\{\rv{U}_t \},\rv{U}_i- g(\rv{V}_{i-L}^{i-1}) ;\rv{V}_i -g(\rv{V}_{i-L}^{i-1})| \rv{V}_{i-L}^{i-1}) \\
	&= I(\{\rv{U}_t \},\rv{U}_i- \hat{\rv{U}}_i^{(L)} ;\rv{V}_i - \hat{\rv{U}}_i^{(L)}| \rv{V}_{i-L}^{i-1}) \\
	&= I(\{\rv{U}_t \},\rv{W}_i^{(L)} ;\rv{Q}_i^{(L)}| \rv{V}_{i-L}^{i-1}) \\
	&= I(\{\rv{U}_t \},\rv{W}_i^{(L)} ;\rv{W}_i^{(L)} + \rv{N}_i| \rv{V}_{i-L}^{i-1}) \\
	&= I(\rv{W}_i^{(L)} ;\rv{W}_i^{(L)} + \rv{N}_i| \rv{V}_{i-L}^{i-1}) \\
	&= I(\rv{W}_i^{(L)} ;\rv{W}_i^{(L)} + \rv{N}_i) \\
	&= \frac{1}{2} \log \left(1+ \frac{\sigma_L^2}{\theta}\right).
\end{align}
Taking the limit as $ L \rightarrow \infty $ we obtain:
\begin{align}
	I(\{\rv{U}_t \} ;\rv{V}_i | \rv{V}_{i}^{-})  
	&= \lim_{L \rightarrow \infty } I(\{\rv{U}_t \} ;\rv{V}_i | \rv{V}_{i-L}^{i-1}) \\
	&= \frac{1}{2} \log \left(1+ \frac{\sigma_\infty^2}{\theta}\right) \\
	&= I(\rv{W}_n;\rv{W}_n + \rv{N}_n).
\end{align}
Further note
\begin{align}
	\bar{I}(\{\rv{U}_t\}, \{\rv{V}_t\})
	&= \lim_{n \rightarrow \infty} \frac{1}{n} I(\rv{U}^n;\rv{V}^{n})  \\
	&= \lim_{n \rightarrow \infty} \frac{1}{n} \sum_{i=1}^{n} I(\rv{U}^n;\rv{V}_i|\rv{V}^{i-1})  \\
	&= I(\{\rv{U}_t \} ;\rv{V}_i | \rv{V}_{i}^{-})  \\
	&=  I(\rv{W}_n;\rv{W}_n + \rv{N}_n) \\
	&= \frac{1}{2} \log \left(1+ \frac{\sigma_\infty^2}{\theta}\right) .
\end{align}
Note that by \autoref{lemma:mutual_information_rate_gaussian_processes} and \eqref{eq:ib_rate_linear_awgn_form}, we have
\begin{align}
	&\mathcal{C}(\theta) = \bar{I}(\{\rv{Y}_t\}; \{ \rv{Z}_t \}) \\
	&= \frac{1}{2} \mkern-5mu \int \mkern-5mu \log \mkern-5mu \left[\frac{\substack{|H_2(f)|^2 |H_1(f)|^2 |G(f)|^2 (1+\Gamma(f)) \\+ |H_2(f)|^2  S_{\rv{N}}(f)   } }{|H_2(f)|^2  S_{\rv{N}}(f)  }\right] \mathrm{d} f \\
	&= \frac{1}{2} \int \log \left(1 + \frac{
		\left(1 - \frac{D_{\theta}(f)}{\Gamma(f)}\right) \frac{\Gamma(f)}{1+ \Gamma(f)} (1+\Gamma(f)) }{\theta  }\right) \mathrm{d} f \\
	&= \frac{1}{2} \int \log \left( \frac{\Gamma(f)}{ D_{\theta}(f)  }\right) \mathrm{d} f.
\end{align}
Furthermore, since $ \rv{U}_n $ is a function of $ \rv{Y}_n $ and since the postfilter $ H_2 $ is invertible , we have
\begin{equation}
	\mathcal{C}(\theta) =\bar{I}(\{\rv{Y}_t\}; \{ \rv{Z}_t \}) \ = \bar{I}(\{\rv{U}_n\}, \{\rv{V}_n\}).
\end{equation}

Similarly,
\begin{align}
	&\mathcal{R}^{ib}(\theta) = \bar{I}(\{\rv{X}_t\}; \{ \rv{Z}_t \}) \\
	&= \frac{1}{2} \int \log \left( \frac{\substack{|H_2(f)|^2 |H_1(f)|^2 |G(f)|^2 (1+\Gamma(f)) \\+ |H_2(f)|^2  S_{\rv{N}}(f)    }}{\substack{|H_2(f)|^2 |H_1(f)|^2 |G(f)|^2  \\+ |H_2(f)|^2  S_{\rv{N}}(f)  }}\right) \mathrm{d} f \\
	&= \frac{1}{2} \int \log \left( \frac{ \left(1- \frac{D_{\theta}(f)}{\Gamma(f) }\right) \Gamma(f) +  \theta    }{\left(1- \frac{D_{\theta}(f)}{\Gamma(f) }\right)\frac{\Gamma(f) }{1+\Gamma(f)} + \theta  }\right) \mathrm{d} f \\
	&= \frac{1}{2} \int \log \left( \frac{ 1+\Gamma(f)      }{1 + D_{\theta}(f)  }\right) \mathrm{d} f .
\end{align}
This completes the proof of \autoref{theorem:prediction_gaussian_ib}.
	
	\subsection{Proof of \acrshort{pf} for Gaussian Process \autoref{lemma:pf_vector_gaussian}}
	Suppose that $\bv{Y}$ and $\bv{Z}$ are jointly Gaussian vectors with covariance matrix $\Sigma_{\bv{Z}\bv{Y}}$, then there exists $\bv{V} \sim \mathcal{N}(0,\Sigma_{\bv{V}})$ with $\Sigma_{\bv{V}} = \Sigma_{\bv{Z}} - \Sigma_{\bv{Z}\bv{Y}} \Sigma_{\bv{Y}}^{-1} \Sigma_{\bv{Z}\bv{Y}}^T $ such that
	$\bv{Z} = \Sigma_{\bv{Z}\bv{Y}} \Sigma_{\bv{Y}}^{-1} \bv{Y} + \bv{V}$.
Let $U_2^T \Gamma V_2$ be the \acrshort{svd}
of $\Sigma_{\bv{Z}} ^{-1/2} \Sigma_{\bv{Z} \bv{Y}} \Sigma_{\bv{Y}} ^{-1/2} $, where $U_2$ and $V_2$ are two orthogonal matrices and $\Gamma$ is a diagonal matrix with singular values on the diagonal.

We further define the following transformations $\tilde{\bv{Z}} = \tilde{T}_z \bv{Z}$ and $\tilde{\bv{Y}} = \tilde{T}_y \bv{Y}$, where $\tilde{T}_z = U_2 \Sigma_{\bv{Z}} ^{-1/2}$ and $\tilde{T}_y = V_2 \Sigma_{\bv{Y}} ^{-1/2}$. Note that
\begin{align}
    \Sigma_{\tilde{\bv{Z}}} &= \tilde{T}_z \Sigma_{\bv{Z}} \tilde{T}_z^T = I_{n_z}, \\
    \Sigma_{\tilde{\bv{Y}}} &= \tilde{T}_y \Sigma_{\bv{Y}} \tilde{T}_y^T = I_{n_y}, \\ 
    \Sigma_{\tilde{\bv{Z}} \tilde{\bv{Y}}} &= \tilde{T}_z \Sigma_{\bv{Z} \bv{Y}} \tilde{T}_y^T = U_2 \Sigma_{\bv{Z}} ^{-1/2} \Sigma_{\bv{Z} \bv{Y}} \Sigma_{\bv{Y}} ^{-1/2} V_2^T = \Gamma.
\end{align}

We are interested in the PF optimization problem from \eqref{eq:normalized_pf_function},
which is a minimization of convex function over the complement of an open convex set, therefore the minimum is obtained on the boundary of the set.

Since $\bv{Y}$ and $\bv{X}$ are jointly Gaussian, there exists $\bv{W} \sim \mathcal{N}(0,\Sigma_{\bv{Y}} - \Sigma_{\bv{Y}\bv{X}} \Sigma_{\bv{X}}^{-1} \Sigma_{\bv{Y}\bv{X}}^T)$ such that
    $\bv{Y} = \Sigma_{\bv{Y}\bv{X}} \Sigma_{\bv{X}}^{-1} \bv{X} + \bv{W}$.
Furthermore, considering the singular value decomposition of
$\Sigma_{\bv{Y}}^{-1/2}\Sigma_{\bv{Y}\bv{X}} \Sigma_{\bv{X}}^{-1/2} = U_1^T \Lambda V_1$,
the rate constraint obtains the following form:
\begin{align}
    I(\bv{Y};\bv{X}) 
    &= h(\bv{Y}) - h(\bv{W}) \\
    &= \frac{1}{2} \log \frac{\det \Sigma_{\bv{Y}}}{\det (\Sigma_{\bv{Y}} - \Sigma_{\bv{Y}\bv{X}} \Sigma_{\bv{X}}^{-1} \Sigma_{\bv{Y}\bv{X}}^T) } \\
    &=- \frac{1}{2} \log \det (I - \Lambda^2) \\
    &=-\sum_{i=1}^n \frac{1}{2} \log(1-\lambda_i^2),
\end{align}
where we identify $C_{1i} \triangleq - \frac{1}{2} \log(1-\lambda_i^2) $.
Next, consider the objective function. Note that
\begin{equation}
    \bv{Z} = \Sigma_{\bv{Z}\bv{Y}} \Sigma_{\bv{Y}}^{-1} \Sigma_{\bv{Y}\bv{X}} \Sigma_{\bv{X}}^{-1} \bv{X} + \Sigma_{\bv{Z}\bv{Y}}  \Sigma_{\bv{Y}}^{-1} \bv{W} + \bv{V},
\end{equation}
and so,
\begin{align}
    &I(\bv{Z};\bv{X}) 
    = h(\bv{Z}) - h( \Sigma_{\bv{Z}\bv{Y}} \Sigma_{\bv{Y}}^{-1} \bv{W} + \bv{V}) \\
    &=\frac{1}{2} \log \frac{\det \Sigma_{\bv{Z}}}{\det( 
    \Sigma_{\bv{Z}} -
    \Sigma_{\bv{Z}\bv{Y}}  \Sigma_{\bv{Y}}^{-1}  \Sigma_{\bv{Y}\bv{X}} \Sigma_{\bv{X}}^{-1} \Sigma_{\bv{Y}\bv{X}}^T
     \Sigma_{\bv{Y}}^{-1} \Sigma_{\bv{Z}\bv{Y}}^T
    )}\\
    &= - \frac{1}{2} \log \det( I
     - U_2^T \Gamma V_2 U_1^T \Lambda V_1 V_1^T \Lambda U_1 V_2^T \Gamma U_2^T
    ) \\
    &= - \frac{1}{2} \log \det( I
     - V_2^T \Gamma^2 V_2 U_1^T \Lambda^2 U_1 
    ).
\end{align}
This completes the proof of \autoref{lemma:pf_vector_gaussian}.
	
	\subsection{Proof of Robust Information Bottleneck for Gaussian Processes from \autoref{theorem:comib_discrete_gaussian_RP}}
	
Note that for some $ \pmf{\rv{X}^n\rv{Y}^n} $ that satisfies $ I(\rv{X}^n;\rv{Y}^n) \geq C_1 $, it holds that
\begin{equation}
	R_n^{comib} (C_1,C_2) \leq R_n^{ib}(C_2;\pmf{\rv{X}^n\rv{Y}^n}).
\end{equation}
Thus,
\begin{align}
	\mathcal{R}^{comib}(C_1,C_2) 
	&= \lim_{n \rightarrow \infty} R_n^{comib} (C_1,C_2) \\
	&\leq \lim_{n \rightarrow \infty} R_n^{ib}(C_2;\pmf{\rv{X}^n\rv{Y}^n}) \\
	&= \mathcal{R}^{ib}(C_2; \Gamma(f)). 
\end{align}

First direction. Assume the \acrshort{snr} spectrum is a constant, i.e., $ \Gamma(f) = \psi  $. Since $ \bar{I}(\{\rv{X}_n\}, \{\rv{Y}_n\}) = C_1 $, $ \psi $ is the solution to
\begin{align}
    C_1 &= \frac{1}{2} \int_{-1/2}^{1/2} \log \left(\frac{|H(f)|^2 S_{\rv{X}}(f) + S_{\rv{W}}(f)}{S_{\rv{W}}(f)}\right) \mathrm{d}f \\
    &= \frac{1}{2} \int_{-1/2}^{1/2} \log \left(1 + \Gamma(f)\right) \mathrm{d}f \\
    &= \frac{1}{2} \log(1+\psi),
\end{align}
thus, $ \psi = 2^{2 C_1}-1 $.  Furthermore, 
\begin{align}
	C_2 &= \frac{1}{2} \int_{-1/2}^{1/2} \log \left(\frac{\psi}{\theta}\right)\mathrm{d}f \\
	&=  \frac{1}{2} \log\frac{\psi}{\theta},
\end{align}
then,
\begin{equation}
	\theta = \psi \cdot 2^{-2C_2} = ( 2^{2C_1}-1) 2^{-2C_2}.
\end{equation}
This implies, together with \autoref{theorem:ib_gaussian_process}, that the \acrshort{comib} rate is upper bounded by
\begin{align} 
    R^{comib} (C_1,C_2) 
    &\leq  \frac{1}{2}\log \left(\frac{2^{C_1}}{1+( 2^{2C_1}-1) 2^{-2C_2}}\right) \\
    &=  -\frac{1}{2} \log [1-(1-2^{-2C_1})(1-2^{-2C_2})]. \label{eq:comib_rate_upper_bound}
\end{align}

Similarly, for some $ \pmf{\rv{Y}^n\rv{Z}^n} $,  consider the normalized \acrshort{pf} function from \eqref{eq:normalized_pf_function}.
It holds that, for some $ \pmf{\rv{Y}^n\rv{Z}^n} $ that satisfies $ I(\rv{Y}^n;\rv{Z}^n) \leq C_2 $
\begin{equation}
	R_n^{comib} (C_1,C_2) \geq R_n^{pf}(C_1;\pmf{\rv{X}^n\rv{Y}^n}).
\end{equation}
For some fixed $ n $, denote $ \bv{X} \triangleq \rv{X}^n $, $ \bv{Y} \triangleq \rv{Y}^n $ and $ \bv{Z} \triangleq \rv{Z}^n $.
For the opposite direction, i.e., assuming the channel from $\rv{Y}$ to $\rv{Z}$ is \acrshort{awgn}, we utilize \autoref{lemma:pf_vector_gaussian} and establish that the \acrshort{snr} spectrum of $(\bv{X},\bv{Y})$ is also white.

Suppose that the channel from $\rv{Y}_n $ to $\rv{Z}_n $ is an additive Gaussian noise channel with white \acrshort{snr} spectrum, i.e. 
\begin{equation}
    \rv{Z}_n = g_{n} *\rv{Y}_n +  \rv{V}_n,
\end{equation}
where 
\begin{align}
	S_{\rv{Z}}(f)
	&= |G(f)|^2 S_{\rv{Y}}(f) + S_{\rv{V}}(f) \\
	&= \frac{|G(f)|^2 S_{\rv{Y}}(f) + S_{\rv{V}}(f)}{S_{\rv{V}}(f)} \cdot S_{\rv{V}}(f).
\end{align}
We denote the \acrshort{snr} spectrum of the channel $ \rv{Y}_n \rightarrow \rv{Z}_n $ by
\begin{equation}
	\Lambda(f) \triangleq \frac{|G(f)|^2 S_{\rv{Y}}(f) + S_{\rv{V}}(f)}{S_{\rv{V}}(f)}.
\end{equation}
In particular, we assume here that $ \Lambda(f) $ is white, therefore, there exists a constant $ \lambda $ such that $ \Lambda(f) = \lambda $. Applying \autoref{lemma:mutual_information_rate_gaussian_processes}, $ \lambda $ is the solution to the following equation:
\begin{align}
	C_2 
	&= I(\{\rv{Y}_t\}; \{\rv{Z}_t\}) \\
	&=  \frac{1}{2} \log (1+\lambda),
\end{align}
implying  $ \lambda = 2^{2C_2} -1$. 
%
Due to \autoref{lemma:equivalence_of_measures} we have 
\begin{equation}
	\psi_i^2 = \frac{\lambda}{1+\lambda} = \psi^2,
\end{equation}
thus, $ \Psi = \psi \cdot I $,
and \eqref{eq:pf_rate_general_gaussian} becomes
\begin{equation}
	\begin{aligned}
		&R_n^{PF}(C_1) =
		&& \underset{\{\phi_{i}\}}{\text{minimize}}
		&& -\frac{1}{2n} \sum_{i=1}^n \log (1-\psi^2 \phi_{i}^2) \\
		&&& \text{subject to}
		&& -\frac{1}{2n} \sum_{i=1}^n  \log (1-\phi_{i}^2) \geq C_1.
	\end{aligned}
\end{equation}

Denoting $ \zeta_i \triangleq \phi_i^2 $, we obtain the respective Lagrangian
\begin{multline}
	L(\{\zeta_{i}\}, \mu)  \\
	= -\frac{1}{2n} \sum_{i=1}^n  \left[ \log (1-\psi^2 \zeta_{i}) - \mu \log (1-\zeta_{i}) \right] + \mu C_1.
\end{multline}
The KKT conditions are given by:
\begin{itemize}
	\item Stationarity:
	\begin{equation}
		\frac{\partial L}{\partial \zeta_i} = \frac{1}{2n} \left[  \frac{ \psi^2}{1-\psi^2 \zeta_{i}} - \frac{\mu}{1-\zeta_i } \right] = 0,
	\end{equation}
	which implies
	\begin{equation}
		\zeta_i^* = \frac{\psi^2 + \mu }{\psi^2(1+\mu)};
	\end{equation}
	\item Complementary Slackness:
	\begin{equation}
		\mu \left[C_1- \frac{1}{2n}  \sum_{i=1}^n\log (1-\zeta_{i}^*) \right]= 0.
	\end{equation}
\end{itemize}
Thus $ \zeta_{i}^* = \zeta $  for every $ i \in [n] $. In particular,
\begin{equation}
	C_1 = - \frac{1}{2n}  \sum_{i=1}^n \log (1-\zeta) \rightarrow \zeta = \phi^2 = 1-2^{-2C_1}.
\end{equation}
Thus,
\begin{multline}
	R_n^{comib} (C_1,C_2) \geq R_n^{pf}(C_1) \\
	 = -\frac{1}{2} \log [1-(1-2^{-2C_1})(1-2^{-2C_1})],
\end{multline}
 and 
\begin{equation} \label{eq:comib_rate_lower_bound}
	\mathcal{R}^{comib}(C_1,C_2) \geq -\frac{1}{2} \log [1-(1-2^{-2C_1})(1-2^{-2C_2})].
\end{equation}
The proof of \autoref{theorem:comib_discrete_gaussian_RP} then follows from \eqref{eq:comib_rate_lower_bound}, \eqref{eq:comib_rate_upper_bound} and the saddle point property from \autoref{lemma:optimality_saddle_point}.

Thus,
\begin{equation} 
	\mathcal{R}^{comib}(C_1,C_2) = -\frac{1}{2} \log [1-(1-2^{-2C_1})(1-2^{-2C_2})].
\end{equation}
This rate can be achieved by choosing $ \Gamma(f) = \gamma = 2^{2C_1}-1 $ and $ \Lambda(f) = \lambda = 2^{2C_2}-1 $.

Since
\begin{equation}
	\rv{Y}_n = h_n * \rv{X}_n + \rv{W}_n, 
\end{equation}
then,
\begin{equation}
	H(f) = \frac{S_{\rv{Y}\rv{X}}(f)}{S_{\rv{X}}(f)}, \quad S_{\rv{W}}(f) = S_{\rv{Y}}(f) - \frac{|S_{\rv{Y}\rv{X}}(f)|^2}{S_{\rv{X}}(f)}.
\end{equation}
Thus
\begin{equation}
	\Gamma(f) = \frac{|S_{\rv{Y}\rv{X}}(f)|^2}{S_{\rv{X}}(f)S_{\rv{Y}}(f) - |S_{\rv{Y}\rv{X}}(f)|^2} = \gamma,
\end{equation}
implies
\begin{equation}
	|H(f)|^2 = \frac{\gamma}{1+\gamma} \cdot \frac{S_{\rv{Y}}(f) }{S_{\rv{X}}(f)}, \quad S_{\rv{W}}(f) = \frac{1}{1+\gamma} S_{\rv{Y}}(f) .
\end{equation}
As for 
\begin{equation}
	\rv{Z}_n = g_n * \rv{Y}_n + \rv{V}_n
\end{equation}
we may choose $ S_{\rv{Z}} ( f) = S_{\rv{Y}}(f) $, and
\begin{equation}
	|G(f)|^2 = \frac{\lambda}{1+\lambda}, \quad S_{\rv{V}}(f) = \frac{1}{1+\lambda} \cdot S_\rv{Y}(f).
\end{equation}
This completes the proof of \autoref{theorem:comib_discrete_gaussian_RP}.
}

\newpage
\bibliographystyle{IEEEtran}
	\bibliography{prediction}

\begin{thebibliography}{10}
\providecommand{\url}[1]{#1}
\csname url@samestyle\endcsname
\providecommand{\newblock}{\relax}
\providecommand{\bibinfo}[2]{#2}
\providecommand{\BIBentrySTDinterwordspacing}{\spaceskip=0pt\relax}
\providecommand{\BIBentryALTinterwordstretchfactor}{4}
\providecommand{\BIBentryALTinterwordspacing}{\spaceskip=\fontdimen2\font plus
\BIBentryALTinterwordstretchfactor\fontdimen3\font minus
  \fontdimen4\font\relax}
\providecommand{\BIBforeignlanguage}[2]{{%
\expandafter\ifx\csname l@#1\endcsname\relax
\typeout{** WARNING: IEEEtran.bst: No hyphenation pattern has been}%
\typeout{** loaded for the language `#1'. Using the pattern for}%
\typeout{** the default language instead.}%
\else
\language=\csname l@#1\endcsname
\fi
#2}}
\providecommand{\BIBdecl}{\relax}
\BIBdecl

\bibitem{Tishby1999}
N.~Tishby, F.~C.~N. Pereira, and W.~Bialek, ``The information bottleneck
  method,'' in \emph{Proc. 37th Annu. Allerton Conf. Commun. Control Comput.},
  Sep. 1999, p. 368–377.

\bibitem{Cover2006}
T.~M. Cover and J.~A. Thomas, \emph{Elements of {I}nformation {T}heory}.\hskip
  1em plus 0.5em minus 0.4em\relax Hoboken, NJ, USA: Wiley, 2006.

\bibitem{Dobrushin1962}
R.~{Dobrushin} and B.~{Tsybakov}, ``Information transmission with additional
  noise,'' \emph{IRE Trans. Inf. Theory}, vol.~8, no.~5, pp. 293--304, Sep.
  1962.

\bibitem{Wolf1970}
J.~Wolf and J.~Ziv, ``Transmission of noisy information to a noisy receiver
  with minimum distortion,'' \emph{IEEE Trans. Inf. Theory}, vol.~16, pp.
  406--411, Jul. 1970.

\bibitem{Zaidi2020}
A.~Zaidi, I.~E. Aguerri, and S.~S. (Shitz), ``On the information bottleneck
  problems: Models, connections, applications and information theoretic
  views,'' \emph{Entropy}, vol.~22, no.~2, p. 151, Feb. 2020.

\bibitem{Chechik2005}
G.~Chechik, A.~Globerson, N.~Tishby, and Y.~Weiss, ``Information bottleneck for
  {G}aussian variables,'' \emph{J. Mach. Learn. Res.}, vol.~6, pp. 165--188,
  Dec. 2005.

\bibitem{Blahut1972}
R.~Blahut, ``Computation of channel capacity and rate-distortion functions,''
  \emph{IEEE Trans. Inf. Theory}, vol.~18, no.~4, pp. 460--473, Jul. 1972.

\bibitem{Arimoto1972}
S.~Arimoto, ``An algorithm for computing the capacity of arbitrary discrete
  memoryless channels,'' \emph{IEEE Trans. Inf. Theory}, vol.~18, pp. 14--20,
  Jan. 1972.

\bibitem{Berger1971}
T.~Berger, \emph{\BIBforeignlanguage{eng}{Rate distortion theory : {A}
  mathematical basis for data compression.}}, Englewood Cliffs, N.J, 1971.

\bibitem{Hirt1988}
W.~Hirt and J.~Massey, ``Capacity of the discrete-time {G}aussian channel with
  intersymbol interference,'' \emph{IEEE Trans. Inf. Theory}, vol.~34, no.~3,
  pp. 38--38, May 1988.

\bibitem{Shamai1991}
S.~Shamai, L.~Ozarow, and A.~Wyner, ``Information rates for a discrete-time
  {G}aussian channel with intersymbol interference and stationary inputs,''
  \emph{IEEE Trans. Inf. Theory}, vol.~37, no.~6, pp. 1527--1539, Nov. 1991.

\bibitem{Gray2006}
R.~M. Gray, ``Toeplitz and circulant matrices: A review,'' \emph{Foundations
  and Trends® in Communications and Information Theory}, vol.~2, no.~3, pp.
  155--239, 2006.

\bibitem{Cioffi1995}
J.~Cioffi, G.~Dudevoir, M.~Vedat~Eyuboglu, and G.~Forney, ``{MMSE}
  decision-feedback equalizers and coding. {I. E}qualization results,''
  \emph{IEEE Trans. Commun.}, vol.~43, no.~10, pp. 2582--2594, Oct. 1995.

\bibitem{Zamir2008}
R.~Zamir, Y.~Kochman, and U.~Erez, ``Achieving the {G}aussian rate–distortion
  function by prediction,'' \emph{IEEE Trans. Inf. Theory}, vol.~54, no.~7, pp.
  3354--3364, Jul. 2008.

\bibitem{Dikshtein2022a}
\BIBentryALTinterwordspacing
M.~Dikshtein, N.~Weinberger, and S.~{Shamai (Shitz)}, ``The compound
  information bottleneck outlook,'' \emph{CoRR}, vol. abs/2205.04567v1, 2022.
  [Online]. Available: \url{https://arxiv.org/abs/2205.04567v1}
\BIBentrySTDinterwordspacing

\bibitem{Globerson2004}
A.~Globerson and N.~Tishby, ``On the optimality of the {G}aussian information
  bottleneck curve,'' \emph{The Hebrew University of Jerusalem, Tech. Rep},
  p.~22, 2004.

\bibitem{Winkelbauer2014}
A.~Winkelbauer, S.~Farthofer, and G.~Matz, ``The rate-information trade-off for
  {G}aussian vector channels,'' in \emph{Proc. IEEE Int. Symp. Inf. Theory},
  Jun. 2014, pp. 2849--2853.

\bibitem{Meidlinger2014}
M.~Meidlinger, A.~Winkelbauer, and G.~Matz, ``On the relation between the
  {G}aussian information bottleneck and {MSE}-optimal rate-distortion
  quantization,'' in \emph{Proc. IEEE SSP}, Jun. 2014, pp. 89--92.

\bibitem{Homri2018}
A.~{Homri}, M.~{Peleg}, and S.~{Shamai Shitz}, ``Oblivious
  fronthaul-constrained relay for a {G}aussian channel,'' \emph{IEEE Trans.
  Commun.}, vol.~66, no.~11, pp. 5112--5123, Nov. 2018.

\bibitem{Makhdoumi2014}
A.~Makhdoumi, S.~Salamatian, N.~Fawaz, and M.~Médard, ``From the information
  bottleneck to the privacy funnel,'' in \emph{Proc. IEEE Inf. Theory Workshop
  (ITW)}, Nov. 2014, pp. 501--505.

\bibitem{Goldfeld2020}
Z.~{Goldfeld} and Y.~{Polyanskiy}, ``The information bottleneck problem and its
  applications in machine {L}earning,'' \emph{IEEE J. Sel. Areas Inf. Theory},
  vol.~1, no.~1, pp. 19--38, May 2020.

\bibitem{Gibson1998}
Gibson, \emph{\BIBforeignlanguage{eng}{Digital compression for multimedia :
  principles and standards / Jerry D. Gibson ... [et al.].}}\hskip 1em plus
  0.5em minus 0.4em\relax San Francisco, Calif: Morgan Kaufmann, 1998.

\bibitem{zamir2014}
R.~Zamir, B.~Nazer, Y.~Kochman, and I.~Bistritz, \emph{Lattice Coding for
  Signals and Networks: A Structured Coding Approach to Quantization,
  Modulation and Multiuser Information Theory}.\hskip 1em plus 0.5em minus
  0.4em\relax Cambridge University Press, 2014.

\bibitem{Guess2005}
T.~Guess and M.~Varanasi, ``An information-theoretic framework for deriving
  canonical decision-feedback receivers in gaussian channels,'' \emph{IEEE
  Trans. Inf. Theory}, vol.~51, no.~1, pp. 173--187, Jan. 2005.

\bibitem{Boyd2014}
S.~P. Boyd and L.~Vandenberghe, \emph{{Convex Optimization}}.\hskip 1em plus
  0.5em minus 0.4em\relax New York, NY, USA.: Cambridge University Press, 2014.

\bibitem{Shamai2021}
S.~Shamai, ``The information bottleneck: A unified information theoretic
  view,'' National Conference on Communications (NCC2021), Jul. 2021, {P}lenary
  Address.

\bibitem{Guo2013}
D.~Guo, S.~{Shamai (Shitz)}, and S.~Verdú, ``The interplay between information
  and estimation measures,'' \emph{Found. Trends Signal Process.}, vol.~6,
  no.~4, pp. 243--429, 2012.

\bibitem{Dembo1991}
A.~Dembo, T.~Cover, and J.~Thomas, ``Information theoretic inequalities,''
  \emph{IEEE Trans. Inf. Theory}, vol.~37, no.~6, pp. 1501--1518, Nov. 1991.

\bibitem{Horn2012}
R.~A. Horn, \emph{\BIBforeignlanguage{eng}{Matrix Analysis}}, 2nd~ed.\hskip 1em
  plus 0.5em minus 0.4em\relax Cambridge: Cambridge University Press, 2012.

\bibitem{Komaee2020}
\BIBentryALTinterwordspacing
A.~Komaee, ``Mutual information rate between stationary {G}aussian processes,''
  \emph{Results in Applied Mathematics}, vol.~7, p. 100107, 2020. [Online].
  Available:
  \url{https://www.sciencedirect.com/science/article/pii/S2590037420300182}
\BIBentrySTDinterwordspacing

\end{thebibliography}

\end{document}